\documentclass[11pt]{article} 

\usepackage[margin=1in,letterpaper]{geometry} 
\usepackage[utf8]{inputenc}
\usepackage[T1]{fontenc}
\usepackage{mathpazo} 

\usepackage{authblk,comment,cprotect,xspace}
\usepackage[normalem]{ulem} 

\usepackage{booktabs,subcaption} 

\usepackage{enumitem} 
\usepackage[dvipsnames]{xcolor} 
\definecolor{ptblue}{RGB}{15,76,129} 
\definecolor{ptemerald}{HTML}{009473} 
\definecolor{ptilluminating}{HTML}{F5DF4D} 
\definecolor{ptgray}{HTML}{939597} 

\usepackage{tikz} 
\usetikzlibrary{decorations.pathreplacing}

\usepackage{amsmath,amsfonts,amssymb,amsthm,bbm,mathtools} 

\let\OLDleft\left
\let\OLDright\right
\renewcommand{\left}{\mathopen{}\mathclose\bgroup\OLDleft}
\renewcommand{\right}{\aftergroup\egroup\OLDright}

\let\OLDland\land
\renewcommand{\land}{\:\OLDland\:}
\let\OLDlor\lor
\renewcommand{\lor}{\:\OLDlor\:}
\let\OLDforall\forall
\renewcommand{\forall}{\OLDforall\:}
\let\OLDexists\exists
\renewcommand{\exists}{\OLDexists\,}

\newcommand*{\diff}[1]{\mathop{}\!\mathrm{d}#1} 

\DeclareMathOperator*{\argmin}{arg\,min}

\usepackage[ruled,linesnumbered,vlined]{algorithm2e}

\SetCommentSty{mycommentfont}

\usepackage[square,sort]{natbib} 
\usepackage{hyperref}
\hypersetup{
linktocpage,
colorlinks=true,
citecolor=ptemerald, 
urlcolor=ptblue, 
linkcolor=Plum, 
}
\usepackage{cleveref} 

\theoremstyle{plain}
\newtheorem{theorem}{Theorem}[section]
\newtheorem{claim}[theorem]{Claim}

\newtheorem{corollary}[theorem]{Corollary}
\newtheorem{lemma}[theorem]{Lemma}

\theoremstyle{definition}
\newtheorem{definition}[theorem]{Definition}
\newtheorem{example}[theorem]{Example}

\theoremstyle{remark}
\newtheorem*{remark}{\upshape\bfseries Remark}

\makeatletter 
\newcommand{\EF}[1]{\if\relax\detokenize\expandafter{\@firstofone#1{}}\relax \text{EF}\xspace\else \text{EF#1}\fi}
\makeatother
\newcommand{\EFOne}{\EF{1}\xspace}
\newcommand{\EFX}{\EF{X}\xspace}
\newcommand{\EFM}{\EF{M}\xspace}
\newcommand{\EFXM}{\EF{XM}\xspace}
\newcommand{\EFoneM}{\EF{1M}\xspace} 

\newcommand{\twoAgentTwoThirdsAlloc}{2\textsc{-Agent-}\frac{2}{3}\textsc{-MMS-Alloc}\xspace}
\newcommand{\twoThirdsAlloc}{3\textsc{-Agent-}\frac{2}{3}\textsc{-MMS-Alloc}\xspace}
\newcommand{\highValuedAlg}{\textsc{High-Valued-Alloc}\xspace}
\newcommand{\goodBundle}{reducible bundle\xspace}
\newcommand{\GoodBundle}{\expandafter\MakeUppercase\goodBundle} 

\newcommand{\alloc}{\mathcal{A}}
\newcommand{\length}[1]{\texttt{len}(#1)}
\newcommand{\indicator}[1]{\mathbbm{1}_{\left\{#1\right\}}\xspace}
\newcommand{\MMS}{\text{MMS}}
\newcommand{\utility}[2]{u_{#1}(#2)}
\newcommand{\indGoodsOf}[1]{M_{#1}^{\text{IND}}} 
\newcommand{\divGoodsOf}[1]{M_{#1}^{\text{DIV}}} 
\newcommand{\pieceOfGoodForAgent}[2]{\widetilde{#1}_{#2}} 

\title{Fair Division with Subjective Divisibility}
\author[1]{Xiaohui Bei}
\author[2]{Shengxin Liu}
\author[3]{Xinhang Lu}
\affil[1]{Nanyang Technological University, \nolinkurl{xhbei@ntu.edu.sg}}
\affil[2]{Harbin Institute of Technology, Shenzhen, \nolinkurl{sxliu@hit.edu.cn}}
\affil[3]{UNSW Sydney, \nolinkurl{xinhang.lu@unsw.edu.au}}
\date{}

\begin{document}
\maketitle

\begin{abstract}
The classic fair division problems assume the resources to be allocated are either divisible or indivisible, or contain a mixture of both, but the agents always have a predetermined and uncontroversial agreement on the (in)divisibility of the resources.
In this paper, we propose and study a new model for fair division in which agents have their own \emph{subjective divisibility} over the goods to be allocated.
That is, some agents may find a good to be indivisible and get utilities only if they receive the \emph{whole} good, while others may consider the same good to be divisible and thus can extract utilities according to the fraction of the good they receive.
We investigate fairness properties that can be achieved when agents have subjective divisibility.
First, we consider the \emph{maximin share (MMS) guarantee} and show that the worst-case MMS approximation guarantee is at most~$2/3$ for $n \geq 2$ agents and this ratio is \emph{tight} in the two- and three-agent cases.
This is in contrast to the classic fair division settings involving two or three agents.
We also give an algorithm that produces a $1/2$-MMS allocation for an arbitrary number of agents.
Second, we study a hierarchy of envy-freeness relaxations, including EF1M, EFM and EFXM, ordered by increasing strength.
While EF1M is compatible with non-wastefulness (an economic efficiency notion), this is not the case for EFM, even for two agents.
Nevertheless, an EFXM and non-wasteful allocation always exists for two agents if at most one good is discarded.
\end{abstract}

\section{Introduction}
\label{sec:introduction}

Fair division studies how to allocate scarce resources among interested agents with potentially different preferences in such a way that every agent involved feels that she gets a fair share.
Dating back to the~1940s, \citet{Steinhaus49} formulated and studied how to fairly divide a cake---the problem is commonly known as \emph{cake cutting}, with the cake serving as a metaphor for heterogeneous \emph{divisible} goods such as land or time.
The two most prominent fairness notions in the domain are \emph{proportionality}~\citep{Steinhaus49} and \emph{envy-freeness}~\citep{Foley67}.
An allocation is said to be proportional if every agent receives value at least~$1/n$ of her total value for the grand set of goods (here, $n$ denoting the number of agents), and envy-free (\EF{}) if each agent weakly prefers her own bundle to any other agent's bundle in the allocation.

Looking beyond cake cutting, recently, there has been considerable attention to the allocation of heterogeneous \emph{indivisible} goods such as jewellery, electronics, artworks, and many other common items~\citep{AmanatidisAzBi23,Suksompong21,Suksompong25}.
While both proportionality and envy-freeness can always be satisfied in cake cutting, neither can always be satisfied when dividing indivisible goods.
In order to circumvent this issue, relaxations of the notions have been studied.
A natural alternative to proportionality is the \emph{maximin share (MMS) guarantee}~\citep{Budish11}, which requires that every agent receives value at least their own \emph{maximin share}, i.e., the largest value that the agent can guarantee for herself if she is allowed to partition the goods into $n$ parts and always receives the worst part.
An MMS allocation may not exist, but a constant multiplicative approximation to MMS can always be satisfied~\citep{AkramiGa24,KurokawaPrWa18}.
On the other hand, envy-freeness is often relaxed to \emph{envy-freeness up to one good (\EFOne)}~\citep{Budish11}, which requires that any envy an agent has towards another agent can be eliminated by removing some good from the latter agent's bundle.
An \EFOne allocation always exists~\citep{CaragiannisKuMo19,LiptonMaMo04}.

Recently, \citet{BeiLiLi21} generalized the two aforementioned classic settings and studied the fair allocation of mixed divisible and indivisible goods (henceforth referred to as \emph{mixed goods}).
They introduced a notion called \emph{envy-freeness for mixed goods (\EFM)} which naturally generalizes both \EF{} and \EFOne, and showed the guaranteed existence of an \EFM allocation.
The weaker and stronger variants of envy-freeness relaxations (e.g., \emph{\EFoneM} and \emph{\EFXM}) have also been studied for mixed-goods allocation~\citep{CaragiannisKuMo19,NishimuraSu23}.
Moreover, \citet{BeiLiLu21} investigated the existence, approximation, and computation of MMS allocations in the mixed-goods model.

In all of the three aforementioned models---cake cutting, indivisible-goods allocation, or mixed-goods model---the (in)divisibility of the goods is \emph{objective} and predetermined.
Put differently, all agents agree with each other on whether a good is divisible or not.
In many real-world scenarios, however, agents may have \emph{subjective divisibility} towards the goods:
\begin{itemize}
\item As our first example, consider the allocation of predetermined time slots for venue usage at a university.
A particular time slot may be considered indivisible for some users---for instance, professors usually need a full time slot for their lectures or the final examination of their courses.
In the meanwhile, the same time slot may be viewed as divisible for the student activity groups, because their usage of the venue is flexible and their utilities may simply be proportional to the time booked.

\item Another example is the allocation of computing resources.
Given a computing resource, like a CPU or 1GB of RAM, some computational tasks may find it indivisible because they require the whole resource to execute, while other tasks may be more flexible and can have different levels of performances based on the fraction of the same resource allocated to them.

\item Our last example touches on the division of assets such as land, real estates and business ownership in various contexts such as divorce settlement and inheritance division.
Consider a piece of land as an example.
This is a typical divisible good for agents with a general purpose; but some others may view it as indivisible if they only intend to build a home on the entire land.
Similarly, when dealing with residential real estate, some agents may perceive a property as indivisible since they plan to live in it.
Conversely, others who already own a home may view the same property as a potential source of rental income and therefore are happy to own just a portion of it.
\end{itemize}

This subjective divisibility of agents over the resources brings interesting and challenging characteristics to the classic fair division problems.
For example, with only two agents, the simple \emph{cut-and-choose protocol}, where the first agent partitions the goods into two parts that are as equal as possible from her point of view and let the second agent choose first, is known to provide strong fairness guarantees in different settings.
This is, however, no longer the case when the two agents may disagree with each other on the divisibility of the goods.
Consider the simple example where agents have the following subjective divisibility and utilities:
\begin{center}
\begin{tabular}{@{}c|*{2}{l}@{}}
\toprule
& Good~$1$ & Good~$2$ \\
\midrule
Agent~$1$ & divisible, $1$ & divisible, $1$ \\
Agent~$2$ & indivisible, $1$ & indivisible, $1$ \\
\bottomrule
\end{tabular}
\end{center}
If agent~$1$ partitions the goods into two parts such that each part contains a half of each good, agent~$2$ values either part at zero.
Regardless of agent~$2$'s choice, the resulting allocation is unfair---a \emph{naive} cut-and-choose protocol fails to give any positive approximation to MMS!
In the meanwhile, there is clearly a better allocation: each agent gets a whole good.
The central questions of the paper are
\begin{itemize}
\item how to design fair allocation algorithms that could cope with subjective divisibility over the resources; and
\item how subjective divisibility affects well-known fairness properties such as the MMS guarantee and the (relaxations of) envy-freeness.
\end{itemize}

\subsection{Our Results}

In this paper, we initiate the study of fair division with \emph{subjective divisibility}.
Our model consists of~$n$ agents and~$m$ goods.
For each good, some agents may regard it as indivisible, meaning that they only derive utilities if they receive the whole good, while other agents may regard the good as divisible and derive utilities proportional to the fraction of the good they receive.\footnote{For simplicity, we present our model concerning \emph{homogeneous} divisible goods.
We discuss in \Cref{sec:extension-cake} that our results still hold with heterogeneous divisible goods, i.e., cakes.}

In \Cref{sec:MMS}, we focus on the maximin share (MMS) guarantee, which admits a natural adaptation in our model: when computing an agent's maximin share, the agent may partition her divisible goods into multiple parts across her MMS bundles.
With subjective divisibility, we show that the worst-case MMS approximation guarantee is at most~$2/3$ for $n \geq 2$ agents.\footnote{For two agents, this result is in contrast to the classic fair division settings (i.e., cake cutting and indivisible goods allocation) and the mixed-goods model where the cut-and-choose protocol always guarantees an MMS allocation~\citep[see, e.g.,][]{Procaccia16,BouveretLe16,BeiLiLu21}.
Note also that with indivisible goods, for any number of agents, a $\big( \frac{3}{4} + \frac{3}{3836} \big)$-approximation to MMS is guaranteed~\citep{AkramiGa24}.}
We match the upper bound by presenting algorithms which always produce $\frac{2}{3}$-MMS allocations for the two- and \emph{three-agent} cases.
It is worth noting that, even in the three-agent case, a \emph{tight} MMS approximation ratio has only been observed in settings with further restrictions~\citep[e.g.,][]{FeigeSaTa21,FarhadiGhHa19}.\footnote{\citet{FeigeSaTa21} showed that the ratio~$39/40$ is tight for instances with three agents and nine goods.
\citet{FarhadiGhHa19} adapted the MMS guarantee to a setting where indivisible goods are allocated to agents who have \emph{unequal entitlements}, and showed that the ratio~$1/n$ is tight for the problem.}
More specifically for three agents, when allocating indivisible goods, the current best MMS approximation ratio is~$11/12$ due to \citet{FeigeNo22} and the current best impossibility result is~$39/40$ due to \citet{FeigeSaTa21}.
With arbitrary number of agents, we show that a $\frac{1}{2}$-MMS allocation always exists.

Technique-wise, the most innovative and technically most involved algorithm is arguably the $\frac{2}{3}$-MMS algorithm for three agents.
The intricacy of handling allocations is twofold.
First, a common technique for computing an (approximate) MMS allocation is to reduce the given instance to an \emph{IDO instance} where agents have an identical ordering over the goods while retaining the same set of valuations they had in the original instance.
This technique has been successfully deployed in the literature since the work of \citet{BouveretLe16}.
At a high level, a desired allocation in the IDO instance induces a picking sequence of the agents for the original instance.
The agents pick their favourite good from the remaining goods according to the picking sequence in the original instance, which results in a bundle worth at least what they received in the IDO instance.
This technique seems to no longer work in our setting due to subjective divisibility.
As a result, we make more effort to examine how agents value (some subsets of) the goods in detail.
On top of that, reasoning about agents' subjective divisibility towards specific subsets of goods adds another intricate aspect, because the allocation space is significantly enlarged, which complicates the analysis.

In \Cref{sec:EFM}, we turn our attention to envy-freeness (\EF{}) and consider a hierarchy of its relaxations.
Satisfying envy-freeness alone is trivial: we can divide each good equally into $n$ pieces and give each agent one piece; since each agent values every bundle equally, the allocation is envy-free.
This allocation, however, is not economically efficient, as allocating an agent a fraction of her indivisible good is wasteful.
We thus introduce a very basic economic efficiency notion called \emph{non-wastefulness} which excludes the scenario where agents receive a fraction of their indivisible goods.
We show that non-wastefulness is compatible with \EFoneM---the weakest fairness notion in the hierarchy---based on the idea that any envy from one agent towards the other can be eliminated if we remove a former agent's positively valued indivisible good from the latter agent's bundle.
The stronger notions in the hierarchy are \EFM and \EFXM, based on the following idea: an agent may expect to be EF1 / EFX towards any agent who only receives the former agent's indivisible goods, and envy-free towards the rest.
Perhaps surprisingly, \EFM and non-wastefulness are incompatible, even for two agents.
On the positive side, we show for two agents, by discarding \emph{at most} one good, an \EFXM and non-wasteful allocation (of the remaining goods) always exists and can be computed in polynomial time.

\subsection{Further Related Work}

Our work has mostly been inspired by the papers of \citet{BeiLiLu21}, \citet{BeiLiLi21} and \citet{NishimuraSu23} on fair division of mixed divisible and indivisible goods.
We refer the interested readers to the survey of \citet{LiuLuSu24} for an overview of the most recent developments in this direction.
Note that the assumption of subjective divisibility further generalizes the mixed-goods model and significantly enlarges the space of possible allocations, as a result, their algorithmic results can hardly be directly applied to our model.

\citet{BeiLiLu21} devised an algorithm which produces an $\alpha$-MMS allocation, where~$\alpha \in [1/2, 1]$ is a monotonically increasing function of how agents value divisible goods relative to their maximin share, and adapted the algorithm to provide a better approximation guarantee that can almost match the best possible approximation ratio in indivisible-goods allocation.
In addition, they compared the worst-case and per-instance MMS approximation guarantee with mixed goods to that with only indivisible goods.

\citet{NishimuraSu23} provided a formal proof showing that a maximum Nash welfare allocation is \emph{EF1M} (a weakening of \EFM) for agents with additive utilities, and \emph{EFXM} (a strengthening of \EFM) when agents have binary valuations over the mixed goods.
In addition to EFM existence, \citet{BeiLiLi21} proposed efficient algorithms to compute EFM allocations in special cases and approximate-EFM allocations, followed by preliminary results on the (in)compatibility of EFM and economic efficiency notions based on \emph{Pareto optimality (PO)}.
The quantitative tradeoff between fairness (of being EFM or its variants) and economic efficiency (of using utilitarian welfare) was studied by \citet{LiLiLu24} through the lens of \emph{price of fairness}, which quantifies the efficiency loss due to fairness requirements.
\citet{LiLiLu23} studied the design of \emph{truthful} and EFM mechanisms in the presence of strategic agents.
While truthfulness and EFM are incompatible even for two agents with additive utilities, they designed truthful and EFM mechanisms in several special cases where the expressiveness of preferences are further restricted.
\citet{BhaskarSrVa21} generalized EFM to a \emph{mixed-resources} setting where both goods and chores are to be allocated, and showed an EFM allocation always exists when allocating indivisible items and bad cake and in certain scenarios of allocating indivisible chores and cake.

When allocating indivisible resources, in addition to the work on approximate fairness, there has been work on targeting at \emph{exact} fairness like proportionality and envy-freeness by \emph{sharing} certain indivisible resources among agents~\citep{MisraSe20,SandomirskiySe22}.
These works focused on obtaining fair and economically efficient allocations with the smallest number of resources to be shared between agents and provided algorithmic and/or hardness results.
We remark here that in their model, every agent agrees with each other that the goods are \emph{indivisible}, which is different from ours.

\section{Preliminaries}
\label{sec:preliminaries}

For~$s \in \mathbb{N}$, let $[s] \coloneqq \{1, 2, \dots, s\}$.
Our model includes a set of agents~$N = [n]$ and a set of goods~$M = \{g_1, g_2, \dots, g_m\}$.
We assume without loss of generality that each good~$g_\ell \in M$ is represented by an interval~$[\ell - 1, \ell]$, i.e., the amount of each good is normalized to one.
Given any~$g \in M$, a \emph{piece of good~$g$}, denoted by~$\widetilde{g}$, is a union of finitely many disjoint subintervals of~$g$, i.e., $\widetilde{g} \subseteq g$.
A \emph{bundle}~$B \coloneqq \bigcup_{g \in M} \widetilde{g}$ consists of a (possibly empty) piece of each good~$g \in M$.
We sometimes abuse the terminology by considering a bundle as the set of the (inclusion-maximal) intervals that it contains.
For notational convenience, given~$g \in M$ and a bundle~$B$, denote by~$g^B \coloneqq g \cap B$ the piece of good~$g$ that is contained in~$B$; $g^B = \emptyset$ if bundle~$B$ does not contain any part of good~$g$.
Denote by~$\length{I} \coloneqq b - a$ the length of interval~$I = [a, b]$ and $\length{B} \coloneqq \sum_{I \in B} \length{I}$ the length of bundle~$B$.
Note that it does not matter if the interval is open or closed for its length.

A $k$-partition~$(P_1, P_2, \dots, P_k)$ of goods~$M$ consists of~$k$ bundles such that $\length{P_i \cap P_j} = 0$ for all $i \neq j$ and $\bigcup_{i \in [k]} g^{P_i} = g$ for each~$g \in M$; note that bundle~$P_i$ might be empty.
Generally speaking, each~$P_i$ consists of a (possibly empty) piece of each good~$g \in M$, which then may consist of multiple disjoint subintervals of~$g$.
Let $\Pi_k(M)$ be the set of all $k$-partitions of goods~$M$.
An \emph{allocation}~$\alloc = (A_1, A_2, \dots, A_n) \in \Pi_n(M)$ is an $n$-partition of~$M$ among the agents, where $A_i$ is the bundle allocated to agent~$i \in N$.

\paragraph{Subjective Divisibility and Utilities}
Each agent~$i \in N$ has a non-negative utility~$\utility{i}{g}$ for each good~$g \in M$.
We assume that each good is positively valued by at least one agent.
In our paper, agents have \emph{subjective divisibility} towards the goods: each~$i \in N$ can partition goods~$M$ into two disjoint subsets~$\indGoodsOf{i}$ and~$\divGoodsOf{i}$,\footnote{For the sake of being succinct, we use ``ind'' as a shorthand for ``indivisible'' and ``div'' for ``divisible'' when we specify agents' subjective divisibility towards a good later in our examples.} where~$\indGoodsOf{i}$ contains her indivisible goods and~$\divGoodsOf{i}$ contains her \emph{homogeneous} divisible goods.\footnote{We assume that the agent regards her zero-valued goods as indivisible.
In other words, the agent positively values her divisible goods.
This assumption, first, provides semantic consistency and cognitive ease.
For instance, it is more natural to consistently regard the remaining part of a good in a reduced instance as indivisible for those agents who think the good is indivisible in the original instance.
It also looks weird that a part of an indivisible good could suddenly become divisible.
Second, as we will see in \Cref{def:EFM}, this assumption makes our EFM and EFXM definitions more sensible.
Without it, an agent can simply report a zero-valued good as divisible; then, the agent need be envy-free towards any agents who receive (a part of) the divisible good.}
More specifically, suppose that agent~$i$ gets a piece of good~$\widetilde{g} \subseteq g$ of length~$x$, then
\[
\utility{i}{\widetilde{g}} = x \cdot \utility{i}{g} \cdot \indicator{\left( g \in \indGoodsOf{i} \land x = 1 \right) \lor g \in \divGoodsOf{i}},
\]
where $\indicator{\cdot}$ is the indicator function that is~$1$ when agent~$i$ receives her whole indivisible good~$g$ (i.e., $g \in \indGoodsOf{i} \land x = 1$) or when agent~$i$ regards good~$g$ as divisible (i.e., $g \in \divGoodsOf{i}$).
We assume \emph{additive} utilities, meaning that agent~$i$'s utility for allocation~$\alloc$ is
\[
\utility{i}{\alloc} = \utility{i}{A_i} = \sum_{g \in M} \utility{i}{g^{A_i}} = \sum_{g \in M} \length{g^{A_i}} \cdot \utility{i}{g} \cdot \indicator{\left( g \in \indGoodsOf{i} \land \length{g^{A_i}} = 1 \right) \lor g \in \divGoodsOf{i}}.
\]
An \emph{instance} consists of the goods~$M$, the agents~$N$ and their subjective divisibility and utilities.

\begin{remark}
We will discuss in \Cref{sec:extension-cake} about the setting with \emph{heterogeneous} divisible goods (i.e., \emph{cakes}).
Each agent will be endowed with a \emph{density function} to capture how the agent values different parts of her divisible goods.
Despite this generalization, it is worth noting that agents still have \emph{additive} utilities over divisible resources.

Our subjective divisibility model is general enough to allow more flexible subjective divisibility, e.g., via more complex utility functions.
For example, given a fraction of an good, the utility an agent get can be a convex or concave function with respect to the fraction.
\end{remark}

\medskip

Our goal is to \emph{fairly} allocate the goods despite the agents having their own subjective divisibility.
To this end, we adapt two fairness notions that have been investigated when allocating mixed goods by \citet{BeiLiLi21,BeiLiLu21} to our model.
Specifically, we study the \emph{maximin share (MMS) guarantee} and \emph{envy-freeness for mixed goods (EFM)} in \Cref{sec:MMS,sec:EFM}, respectively.
Additionally, we examine the compatibility between the fairness notions and economic efficiency notions.
A fundamental efficiency notion in the context of fair division is \emph{Pareto optimality}.
Given an allocation~$\alloc = (A_1, A_2, \dots, A_n)$, another allocation $(A'_1, A'_2, \dots, A'_n)$ is a \emph{Pareto improvement} if $u_i(A'_i) \geq u_i(A_i)$ for all agents~$i \in N$ and $u_i(A'_i) > u_i(A_i)$ for some agent~$i \in N$.
An allocation is \emph{Pareto optimal (PO)} if it does not admit a Pareto improvement.

\section{Maximin Share Guarantee}
\label{sec:MMS}

Our focus in this section is on the \emph{MMS guarantee}, which exhibits as a favourable fair-share-based notion in a variety of fair division settings since it was proposed for indivisible-goods allocation~\citep{Budish11}.

\begin{definition}[$\alpha$-MMS]
Recall that $\Pi_n(M)$ is the set of all $n$-partitions of goods~$M$.
The \emph{maximin share (MMS)} of an agent~$i \in N$ is defined as
\[
\MMS_i(n, M) \coloneqq \max_{(P_1, P_2, \dots, P_n) \in \Pi_n(M)} \min_{j \in [n]} \utility{i}{P_j}.
\]
Any partition for which the maximum is attained is called an \emph{MMS partition} of agent~$i$.

An allocation~$(A_1, A_2, \dots, A_n)$ is said to satisfy \emph{$\alpha$-MMS}, for some $\alpha \in [0, 1]$, if for every agent~$i \in N$, $u_i(A_i) \geq \alpha \cdot \MMS_i(n, M)$.
\end{definition}

We will simply write $\MMS_i$ when parameters~$n$ and~$M$ are clear from the context.
We say a $1$-MMS allocation satisfies the MMS guarantee and write MMS as a shorthand for $1$-MMS.
In the following example, we illustrate the maximin share of an agent and an approximate-MMS allocation.

\begin{example}
\label{ex:MMS:medium-valued}
Consider an instance involving three agents~$\{1, 2, 3\}$ and five goods~$\{g_1, g_2, g_3, g_4, g_5\}$.
Each agent regards exactly two goods as divisible.
Specifically, agent~$1$ (resp.,~$2$ and~$3$) regards goods~$g_4, g_5$ (resp.,~$g_1, g_2$ and~$g_1, g_3$) as divisible.
Each agent values each good at~$0.6$.

It can be verified that the maximin share of each agent is~$1$.
We visualize an MMS partition of agent~$1$ below: her three indivisible goods~$g_1, g_2, g_3$ are included in three different bundles, respectively; her two divisible goods~$g_4, g_5$ are divided across the bundles to equalize their values, e.g., a piece of~$g_4$ of value exactly~$0.4$ is added to the first bundle.
The MMS partitions for agents~$2$ and~$3$ are similar, but they may divide the goods differently due to their own subjective divisibility towards the goods.
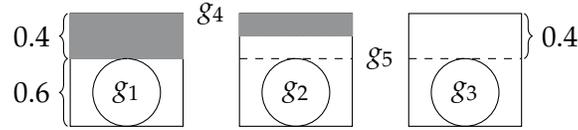
\begin{center}
\begin{tikzpicture}[scale=2]
\draw (0,0) rectangle (1,1);
\draw (1.5,0) rectangle (2.5,1);
\draw (3,0) rectangle (4,1);
\fill[ptgray] (0,.6) rectangle (1,1);
\fill[ptgray] (1.5,.8) rectangle (2.5,1);

\draw (.5,.3) node {$g_1$} circle [radius=0.3];
\draw (2,.3) node {$g_2$} circle [radius=0.3];
\draw (3.5,.3) node {$g_3$} circle [radius=0.3];
\node at (.5,.8) {$g_4$};
\node at (2,.9) {$g_4$};
\node at (2,.7) {$g_5$};
\node at (3.5,.8) {$g_5$};

\draw[dotted,color=ptgray] (0,.6) to (1,.6);
\draw[dotted,color=ptgray] (1.5,.8) to (2.5,.8);
\draw[dashed] (1.5,.6) to (2.5,.6);
\draw[dashed] (3,.6) to (4,.6);

\draw[decorate,decoration={brace,raise=1pt,amplitude=3pt}] (0,0) to (0,.6);
\node[label=left:{$0.6$}] at (0,.3) {};
\draw[decorate,decoration={brace,raise=1pt,amplitude=3pt}] (0,.6) to (0,1);
\node[label=left:{$0.4$}] at (0,.8) {};
\draw[decorate,decoration={brace,raise=1pt,amplitude=3pt}] (1.5,.8) to (1.5,1);
\node[label=left:{$0.2$}] at (1.5,.9) {};
\draw[decorate,decoration={brace,mirror,raise=1pt,amplitude=3pt}] (2.5,.6) to (2.5,.8);
\node[label=right:{$0.2$}] at (2.5,.7) {};
\draw[decorate,decoration={brace,mirror,raise=1pt,amplitude=3pt}] (4,.6) to (4,1);
\node[label=right:{$0.4$}] at (4,.8) {};
\end{tikzpicture}
\end{center}

This instance does not admit an MMS allocation.
A reasonable allocation may be the following:
\begin{itemize}
\item Agent~$1$ receives bundle~$\{g_4, g_5\}$ and gets utility~$0.6 + 0.6 = 1.2$.
\item Agents~$2$ and~$3$ divide the remaining goods evenly, e.g., agent~$2$ receives~$g_2$ and one half of~$g_1$ while agent~$3$ receives~$g_3$ and the other half of~$g_1$.
Each of the agents gets utility~$0.6 + 0.3 = 0.9$.
\end{itemize}
This is a 0.9-MMS allocation.
It can be verified that this allocation also satisfies PO.
\end{example}

It is well-known that an agent's maximin share is NP-hard to compute, even when all goods are objectively indivisible for the agents~\citep{KurokawaPrWa18}.
Nonetheless, computing an agent's MMS value does not need the information of subjective divisibility.
As a result, by utilizing the polynomial-time approximation scheme (PTAS) of \citet[Lemma~4]{BeiLiLu21}, for any~$i \in N$ and constant~$\varepsilon > 0$, we can compute in our setting a partition $(P_1, P_2, \dots, P_n)$ of goods~$M$ such that $\min_{j \in [n]} u_i(P_j) \geq (1 - \varepsilon) \cdot \MMS_i(n, M)$.
For the sake of being self-contained, more details can be found in \Cref{app:MMS-values}.

\bigskip

The remainder of this section is organized as follows.
We start in \Cref{sec:MMS:pre-processing} by providing the core subroutine which handles the allocation of high-valued goods.
Next, as a warm-up, we establish in \Cref{sec:MMS:2-agents} a tight $\frac{2}{3}$-approximation to MMS for two agents.
In \Cref{sec:MMS:3-agents}, we present our main result---a constructive algorithm that always produces a $\frac{2}{3}$-MMS allocation for three agents, matching our impossibility result.
Then, in \Cref{sec:MMS:n-agents}, we devise a $\frac{1}{2}$-MMS algorithm for any number of agents.
Finally, in \Cref{sec:MMS:discussion}, we discuss related considerations of computation and economic efficiency.

\subsection{Pre-processing High-Valued Goods}
\label{sec:MMS:pre-processing}

A core subroutine, which will be utilized in \Cref{alg:2-agent-2/3-MMS,alg:2/3-MMS,alg:n-agent-1/2-MMS} for the two-, three-, and $n$-agent cases presented in this section, allocates ``high-valued goods'' to some agents, and then reduces the given instance to a smaller instance where there is no more high-valued goods.
In other words, if a good is of high value for several agents and these agents regard the good as divisible, dividing this single good among the agents may already satisfy them.
A caveat of reducing an instance to a smaller size is that the remaining agents' MMS in the reduced instance may be lower, in contrast to the case of allocating (objectively) indivisible goods.\footnote{Since most MMS literature on allocating indivisible goods relies on the fact that agents' MMS in a reduced instance (after removing a \emph{single} good) is at least as much as that in the original instance, it is tempting to believe it also holds in our model.
This is, however, not the case if both divisible and indivisible goods are presented, even if the (in)divisibility of the goods is objective as in~\citep{BeiLiLu21}.
We therefore do \emph{not} re-compute the remaining agents' MMS in the reduced instance, and instead always use their MMS computed in the \emph{original} instance throughout the execution of our algorithms.}

\begin{algorithm}[t]
\caption{$\highValuedAlg(\beta, N, M, (\MMS_i)_{i \in N})$}
\label{alg:high-valued}
\DontPrintSemicolon

\KwIn{Parameter~$\beta \in [0, 1]$, agents~$N$, goods~$M$, and maximin share~$(\MMS_i)_{i \in N}$.}
\KwOut{A $\beta$-MMS allocation $(A_i)_{i \in N'}$ for agents~$N' \subseteq N$ and a (reduced) instance with agents~$N'' = N \setminus N'$ and goods~$M'' = M \setminus \bigcup_{i \in N'} A_i$.
}

$\forall i \in N, A_i \gets \emptyset$\;

$N' \gets \emptyset$, $N'' \gets N$, $M'' \gets M$\;

\While{$|N''| \geq 2$ and $\exists a \in N'', g \in M''$ such that $\utility{a}{g} \geq \beta \cdot \MMS_a$}{
	Each agent~$i \in N''$ claims the least fraction of good~$g$, denoted as~$\pieceOfGoodForAgent{g}{i}$, such that $\utility{i}{\pieceOfGoodForAgent{g}{i}} \geq \beta \cdot \MMS_i$, or reports infinity if $\utility{i}{g} < \beta \cdot \MMS_i$.\; \label{ALG-n:high-valued:cut}
	$i^* \gets \argmin_{i \in N''} \length{\pieceOfGoodForAgent{g}{i}}$ \tcp*{Arbitrary tie-breaking.}
	$A_{i^*} \gets \pieceOfGoodForAgent{g}{i^*}$, $g \gets g \setminus \pieceOfGoodForAgent{g}{i^*}$, $N'' \gets N'' \setminus \{i^*\}$, $M'' \gets M'' \setminus \pieceOfGoodForAgent{g}{i^*}$, $N' \gets N' \cup \{i^*\}$\;
}

\If{$|N''| = 1$}{
	Give all remaining goods~$M''$ to the last agent in~$N''$.\;
	$N'', M'' \gets \emptyset$, and $N' \gets N$\;
}

\Return{$(N', (A_i)_{i \in N'}, N'', M'')$}
\end{algorithm}

Let $\beta \in [0, 1]$.
A good is said to be \emph{$\beta$-high-valued} if there exists some agent who values it at least $\beta$ times her MMS (computed in the original instance).
We present \Cref{alg:high-valued}, which outputs a $\beta$-MMS allocation of $\beta$-high-valued goods for (a subset of) the agents and a (reduced) instance in which the remaining agents do not have any $\beta$-high-valued good.

\begin{lemma}
\label{lem:n-agent-high-valued}
Let $\beta \in [0, 1]$.
Given any instance with agents~$N = [n]$, goods~$M$, and the agents' maximin share values~$(\MMS_i)_{i \in N}$,
\Cref{alg:high-valued} computes an allocation~$(A_i)_{i \in N'}$ for agents~$N' \subseteq N$ and a (reduced) instance with agents~$N'' = N \setminus N'$ and goods~$M'' = M \setminus \bigcup_{i \in N'} A_i$ such that the following conditions hold:
\begin{enumerate}[label=(\roman*),ref=condition~(\roman*)]
\item\label{item:beta-MMS}
for each~$i \in N'$, $\utility{i}{A_i} \geq \beta \cdot \MMS_i$;

\item\label{item:k-MMS}
for each~$i \in N''$, $\utility{i}{M''} \geq |N''| \cdot \MMS_i$;

\item\label{item:good-value-UB-beta-MMS}
for each~$i \in N''$ and~$g \in M''$, $\utility{i}{g} < \beta \cdot \MMS_i$.
\end{enumerate}
\end{lemma}

\begin{proof}
\Cref{alg:high-valued} starts by initializing $N'$ (resp.,~$N''$ and~$M''$) to $\emptyset$ (resp.,~$N$ and~$M$).
Clearly, an allocation for the empty set $N' = \emptyset$ satisfies $\beta$-MMS vacuously.
Additionally, $\utility{i}{M''} \geq |N''| \cdot \MMS_i$ for all~$i \in N''$.
If \ref{item:good-value-UB-beta-MMS} holds, we immediately have the desired output satisfying all three conditions.
We therefore assume \ref{item:good-value-UB-beta-MMS} is violated.
In other words, some agent~$a \in N''$ values some good~$g \in M''$ at least~$\beta \cdot \MMS_a$.

By the design of the algorithm, the agent~$i^*$ identified in each iteration of the algorithm receives a piece of good worth at least $\beta \cdot \MMS_{i^*}$.\footnote{Note that \cref{ALG-n:high-valued:cut} uses the ``greater-than-or-equal-to'' sign in that if an agent regards good~$g$ as indivisible, the agent may not be able to claim a piece of good~$g$ worth \emph{exactly} $\beta$ times her MMS; the agent then claims the whole good~$g$ as an eligible piece when \cref{ALG-n:high-valued:cut} executes.
On the other hand, when an agent regards good~$g$ as divisible, she is able to claim a piece worth exactly $\beta$ times her MMS if good~$g$ is valuable enough.
Subjective divisibility plays a role here.}
Next, both the agent and her good are removed from the instance; $N'$, $N''$, and $M''$ are updated accordingly.
When the algorithm terminates and $|N''| \geq 2$, \ref{item:good-value-UB-beta-MMS} is satisfied.
Clearly, we have a $\beta$-MMS allocation for those agents who are removed from the instance in the \verb|while|-loop.

We now proceed to show \ref{item:k-MMS} holds after the last \verb|while|-loop being executed (so $|N''| \geq 1$), that is, $\utility{i}{M''} \geq |N''| \cdot \MMS_i$ for all~$i \in N''$.
Fix any~$j \in N''$ and her MMS partition of the \emph{original} goods~$M$.
There are $|N'|$ iterations in total.
In each of the iterations, \emph{at most} a single good is removed from the instance, and the good may be divisible or indivisible for~$j$.
Let $s$ be the number of iterations that $j$'s indivisible goods are considered and allocated.
Therefore, at most~$s$ many of $j$'s indivisible goods are removed from the instance.
Removing those goods only affects at most~$s$ MMS bundles in her MMS partition, so the agent values the remaining goods for now at least $(n - s) \cdot \MMS_j$.
Next, in the other $|N'| - s$ iterations, the allocated goods are divisible for~$j$.
Since each removed piece~$\pieceOfGoodForAgent{g}{i^*}$ is worth at most $\beta \cdot \MMS_j$, we conclude that
\[
\utility{j}{M''} \geq (n - s) \cdot \MMS_j - (|N'| - s) \cdot \beta \cdot \MMS_j \geq (n - |N'|) \cdot \MMS_j = |N''| \cdot \MMS_j.
\]

If $|N''| = 1$, the last agent receives all remaining goods and gets utility at least her MMS due to the above argument.
It implies that the final output allocation remains $\beta$-MMS.
\end{proof}

In the above proof when we show \ref{item:k-MMS}, in each iteration, (the remaining) agents have their own subjective divisibility over the good (in this iteration) to be divided.
We obtain a succinct argument due to our observation that it suffices to fix any agent in the reduced instance, and consider the number of the agent's indivisible goods being divided in \Cref{alg:high-valued}.

\paragraph{Integral Bundle}
When \Cref{alg:high-valued} (and subsequent algorithms introduced in this section) reduce an instance to a smaller size, a good may have been \emph{partially} allocated to some agent(s) but there is still some remaining part of the good in the reduced instance.
Given any reduced instance with the \emph{current} set of goods~$M'$, we refer to a bundle as an \emph{integral bundle} if it consists of each good~$g \in M'$ in its entirety.

\subsection{Warm-up: Two Agents}
\label{sec:MMS:2-agents}

We show a tight $\frac{2}{3}$-approximate MMS guarantee for two agents, and begin with the upper bound.

\begin{theorem}
\label{thm:2-agent-impossibility}
In the two-agent case, the worst-case MMS approximation guarantee is at most~$2/3$.
\end{theorem}

\begin{proof}
Consider the following instance in which the maximin share of each agent is~$1$.
This can be seen from the fact that for each agent~$i \in N$, $\MMS_i(2, M) \leq \frac{u_i(M)}{3} = 1$, and $\MMS_i(2, M) = 1$ is achieved by partitioning $i$'s two indivisible goods into two bundles and partitioning $i$'s divisible good into the two bundles equally.
\begin{center}
\begin{tabular}{@{}c|*{3}{l}@{}}
\toprule
& $g_1$ & $g_2$ & $g_3$ \\
\midrule
Agent~$1$ & ind, $2/3$ & ind, $2/3$ & div, $2/3$ \\
Agent~$2$ & ind, $2/3$ & div, $2/3$ & ind, $2/3$ \\
\bottomrule
\end{tabular}
\end{center}
We assume without loss of generality that good~$g_1$ is given to agent~$1$ because both agents regard it as indivisible and it is meaningless to split the good between the agents in any allocation.
Next, as the agents have conflicting subjective divisibility towards goods~$g_2$ and~$g_3$, no matter how these goods are allocated, some agent gets a utility of at most~$2/3$.
\end{proof}

Later in \Cref{sec:MMS:n-agents}, we will generalize the example and show that the worst-case MMS approximation guarantee is at most~$2/3$ for $n \geq 2$ agents.

\begin{algorithm}[t]
\caption{$\twoAgentTwoThirdsAlloc(N, M, [~])$}
\label{alg:2-agent-2/3-MMS}
\DontPrintSemicolon

\KwIn{Two agents~$N$, goods~$M$, and the \emph{optional} third argument taking as input $(\MMS_i)_{i \in N}$ computed in the \emph{original} instance.}
\KwOut{A $\frac{2}{3}$-MMS allocation for the two agents~$N$.}

\ForEach{$i \in N$}{
	$A_i \gets \emptyset$\;
	Compute the agent's maximin share $\MMS_i(2, M)$ if the third argument is empty.\; \label{ALG-2:compute-MMS}
}

\eIf{there exists a $\frac{2}{3}$-high-valued good}{
	$(A_1, A_2) \gets \highValuedAlg(2/3, N, M, (\MMS_i)_{i \in N})$ \tcp*{\Cref{alg:high-valued}.}
}(\tcp*[f]{Every~$a \in N$ values every good less than $2/3 \cdot \MMS_a$.}){
	$i^* \gets \argmin_{a \in N} \frac{\utility{a}{M}}{\MMS_a}$ \tcp*{$\utility{i^*}{M} \geq 4/3 \cdot \MMS_{i^*}$. ``Chooser.''}
	$j^* \gets N \setminus \{i^*\}$ \tcp*{$\utility{j^*}{M} \geq 2 \cdot \MMS_{j^*}$. ``Cutter.''}
	Add one good integrally at a time to an empty bundle~$B$ until $\frac{2}{3} \cdot \MMS_{j^*} \leq \utility{j^*}{B} \leq \frac{4}{3} \cdot \MMS_{j^*}$.\; \label{ALG-2:form-integral-bundle}
	Give agent~$i^*$ her preferred bundle between~$B$ and~$M \setminus B$ and agent~$j^*$ the other.\;
}

\Return{$(A_1, A_2)$}
\end{algorithm}

We match the upper bound by showing that a $\frac{2}{3}$-MMS allocation always exists for two agents and can be found via \Cref{alg:2-agent-2/3-MMS}, which follows a \emph{cut-and-choose} framework but need a more refined cut step to ensure that the partition is integral and no value loses due to agents' subjective divisibility.
The algorithm will later be used as a subroutine in \Cref{alg:2/3-MMS} which computes a $\frac{2}{3}$-MMS allocation for three agents, and thus is presented more generally by including the \emph{optional} third argument, which, when specified, takes as input agents' MMS computed in the \emph{original} instance.
Specifically, when \Cref{alg:2-agent-2/3-MMS} takes as input an \emph{original} $2$-agent instance, it computes the two agents' MMS in \cref{ALG-2:compute-MMS} and uses the values accordingly.
When it is invoked as a subroutine in \Cref{alg:2/3-MMS} and takes as input a \emph{reduced} $2$-agent instance, as we noted earlier, it does not re-compute the remaining two agents' MMS given this reduced instance and uses those original MMS values passed to \Cref{alg:2-agent-2/3-MMS}.
Note also that the \verb|if|-statements will never be invoked when processing the reduced instance.

It turns out that \Cref{alg:2-agent-2/3-MMS} has a stronger fairness guarantee for instances without any $\frac{2}{3}$-high-valued good, which will prove useful when analyzing \Cref{alg:2/3-MMS}.

\begin{theorem}
\label{thm:2-agent-2/3-MMS}
For two agents, \Cref{alg:2-agent-2/3-MMS} computes a $\frac{2}{3}$-MMS allocation.

In addition, if an instance does not contain any $\frac{2}{3}$-high-valued good, \Cref{alg:2-agent-2/3-MMS} computes an allocation such that $\utility{i^*}{A_{i^*}} \geq \frac{\utility{i^*}{M}}{2}$, where $i^* = \argmin_{a \in N} \frac{\utility{a}{M}}{\MMS_a}$.
\end{theorem}

\begin{proof}
Given any $2$-agent instance, for any~$a \in N$, we have $\utility{a}{M} \geq 2 \cdot \MMS_a$ by the definition of maximin share.
If an instance has at least one $\frac{2}{3}$-high-valued good, our result immediately holds according to \Cref{lem:n-agent-high-valued}.
Now, consider instances that do not contain any $\frac{2}{3}$-high-valued good, i.e., both agents value each good less than $2/3$ times their own MMS.
Because $\utility{j^*}{M} \geq 2 \cdot \MMS_{j^*}$, \cref{ALG-2:form-integral-bundle} of \Cref{alg:2-agent-2/3-MMS} always succeeds in finding an integral bundle~$B$ such that $\utility{j^*}{B} \in [2/3 \cdot \MMS_{j^*}, 4/3 \cdot \MMS_{j^*}]$, which also means that $\utility{j^*}{M \setminus B} = \utility{j^*}{M} - \utility{j^*}{B} \geq 2/3 \cdot \MMS_{j^*}$.
In other words, both integral bundles~$B$ and~$M \setminus B$ are worth at least $2/3 \cdot \MMS_{j^*}$ for agent~$j^*$.
Since agent~$i^*$ gets her preferred bundle between~$B$ and~$M \setminus B$, she receives a utility of at least $\utility{i^*}{M} / 2 \geq 2/3 \cdot \MMS_{i^*}$.
\end{proof}

At first glance, the fact that the \verb|else|-statements of \Cref{alg:2-agent-2/3-MMS} try hard to find an integral partition between two agents may render the agents' subjective divisibility less important and possibly playing no role.
We would like to note that this is not the case, as the \verb|else|-statements pay particular attention to distinguishing between the ``cutter'' and the ``chooser''.
This asymmetry is, in fact, a result of subjective divisibility.
More specifically, as a subroutine later to be used in the $3$-agent algorithm, the agent who values the remaining goods in a reduced $2$-agent instance less than $2$-MMS has to be the ``chooser''.
This can happen because when we remove the first agent and their corresponding bundle in the original $3$-agent instance, subjective divisibility may prevent us from shrink the bundle further.

\subsection{Three Agents: An Algorithmic Result}
\label{sec:MMS:3-agents}

\begin{algorithm}[t]
\caption{$\twoThirdsAlloc(N, M)$}
\label{alg:2/3-MMS}
\DontPrintSemicolon

\KwIn{Agents~$N = [3]$, goods~$M = \{g_1, g_2, \dots, g_m\}$, subjective divisibility, and utilities.}
\KwOut{A $\frac{2}{3}$-MMS allocation for agents~$N$.}

\ForEach{$i \in N$}{
	$A_i \gets \emptyset$\;
	Compute the agent's maximin share~$\MMS_i(3, M)$.\;
}

\BlankLine

\eIf(\tcp*[f]{Process $\frac{2}{3}$-high-valued goods.}){$\exists a \in N, g \in M, \utility{a}{g} \geq 2/3 \cdot \MMS_a$}{ \label{ALG-3:if-high}
	$(N', (A_i)_{i \in N'}, N'', M'') \gets \highValuedAlg(2/3, N, M, (\MMS_i)_{i \in N})$ \tcp*{\Cref{alg:high-valued}.}
	\lIf(\tcp*[f]{\Cref{alg:2-agent-2/3-MMS}.}){$|N''| = 2$}{$(A_i)_{i \in N''} \gets \twoAgentTwoThirdsAlloc(N'', M'', (\MMS_i)_{i \in N''})$} \label{ALG-3:high:call-2-agent-alg}
}(\tcp*[f]{$\forall a \in N, g \in M, \utility{a}{g} < 2/3 \cdot \MMS_a$.}){
	\eIf{there exists a \goodBundle~$B$ with respect to some agent~$i \in N$}{ \label{ALG-3:if-good-bundle}
		$A_i \gets B$, $N \gets N \setminus \{i\}$, $M \gets M \setminus B$\;
		Invoke $\twoAgentTwoThirdsAlloc(N, M, (\MMS_a)_{a \in N})$ \tcp*{\Cref{alg:2-agent-2/3-MMS}.} \label{ALG-3:good-bundle-call-2-agent-alg}
	}(\tcp*[f]{$\exists o, o' \in M, \forall a \in N, \utility{a}{\{o, o'\}} > \MMS_a$.}){ \label{ALG-3:not-good-bundle}
		$\widehat{M} \gets \{g \in M \mid \forall a \in N, \utility{a}{g} > \frac{\MMS_a}{3}\}$. We have $o, o' \in \widehat{M}$ and assume w.l.o.g.\ that agents~$i, j \in N$ have a \emph{common} divisible good~$d \in \widehat{M}$. \tcp*{$|\widehat{M}| \geq 5$ (\Cref{lem:3-agent-handle-no-good-bundle}).}
		Pick~$o, o'$. If $d \in \{o, o'\}$, pick arbitrarily another good from~$\widehat{M} \setminus \{o, o'\}$. Denote by~$g_1, g_2$ the two picked goods other than~$d$.\;
		Arrange goods~$g_1, d, g_2$ on a line in this order.\;
		Agent~$i$ cuts the line into two parts~$B, B'$ of equal value, i.e., $\utility{i}{B} = \utility{i}{B'}$.\; \label{ALG-3:common-div:cut}
		Give agent~$j$ her preferred part between~$B$ and~$B'$ and agent~$i$ the other.\;
		The third agent~$N \setminus \{i, j\}$ gets all remaining goods~$M \setminus \{g_1, d, g_2\}$.\; \label{ALG-3:common-div:last-agent}
	}
}

\Return{$(A_1, A_2, A_3)$}
\end{algorithm}

We are now ready to present an algorithm which always produces a $\frac{2}{3}$-MMS allocation for three agents.
The pseudocode can be found in \Cref{alg:2/3-MMS}.
Together with our impossibility result (see \Cref{coro:n-agent-impossibility}), this will lead to a \emph{tight} $\frac{2}{3}$-approximation to MMS for three agents.

\begin{theorem}
\label{thm:3-agent-2/3-MMS}
For three agents, \Cref{alg:2/3-MMS} computes a $\frac{2}{3}$-MMS allocation.
\end{theorem}

The high-level idea of \Cref{alg:2/3-MMS} is, given any three-agent instance, to find one by one suitable agents and bundles of goods for them, and iteratively reduce the instance to smaller ones.
On the one hand, if an instance has some $\frac{2}{3}$-high-valued good, we can assign (a subset of) the good to the corresponding agent, and reduce the problem to the two-agent case, in which we are able to find a desired allocation due to \Cref{thm:2-agent-2/3-MMS}.
On the other hand, when the instance does not contain any $\frac{2}{3}$-high-valued good, we follow the same idea, though find a suitable bundle of goods for some agent in a different way.
To this end, we introduce the concept of a \emph{\goodBundle}.

\begin{definition}[\GoodBundle]
\label{def:reducible-bundle}
An integral bundle~$B$ is a \emph{\goodBundle with respect to agent~$i \in N$} if
\begin{itemize}
\item $\utility{i}{B} \geq 2/3 \cdot \MMS_i$,
\item for some~$j \in N \setminus \{i\}$, $\utility{j}{M \setminus B} \geq 2 \cdot \MMS_j$, and
\item for~$k \in N \setminus \{i, j\}$, $\utility{k}{M \setminus B} \geq 4/3 \cdot \MMS_k$.
\end{itemize}
\end{definition}

In words, a reducible bundle is valuable enough to some agent and when removing the bundle of goods from the instance, the remaining goods are valuable enough to other agents.
As we will see shortly in \Cref{lem:3-agent-good-bundle}, if a $3$-agent instance has a \goodBundle, a $\frac{2}{3}$-MMS allocation is guaranteed to exist for the instance.
Given a $3$-agent instance, a \goodBundle, however, need not exist.\footnote{\label{ft:reducible-bundle-not-exist}
For instance, one may verify that the instance presented in \Cref{ex:MMS:medium-valued} does not have any \goodBundle.}
Such a situation constitutes the most difficult part of the algorithm and the proof of \Cref{thm:3-agent-2/3-MMS}.
We now need to take a closer look at agents' subjective divisibility towards valuable goods.
Specifically, we show that in this case there must exist a ``valuable'' \emph{common} divisible good, which allows us to \emph{simultaneously} remove two agents and three goods from the instance.

\subsubsection{Proof of Theorem~\ref{thm:3-agent-2/3-MMS}}

We prove \Cref{thm:3-agent-2/3-MMS} in three sequential steps.
To be more specific, we show \Cref{alg:2/3-MMS} outputs a $\frac{2}{3}$-MMS allocation when
\begin{itemize}
\item an instance contains some $\frac{2}{3}$-high-valued good (\Cref{lem:3-agent-2/3-MMS:high});

\item an instance does not have any $\frac{2}{3}$-high-valued good but has a \goodBundle (\Cref{lem:3-agent-good-bundle});

\item an instance contains neither a $\frac{2}{3}$-high-valued good nor a \goodBundle (\Cref{lem:alg-cover-instances,lem:3-agent-handle-no-good-bundle}).
\end{itemize}

First, we look at the case where the instance has some $\frac{2}{3}$-high-valued good.

\begin{lemma}
\label{lem:3-agent-2/3-MMS:high}
Given a $3$-agent instance that has at least one $\frac23$-high-valued good, \Cref{alg:2/3-MMS} computes a $\frac{2}{3}$-MMS allocation.
\end{lemma}

\begin{proof}
Given the instance, the \verb|if|-condition in \cref{ALG-3:if-high} is evaluated true.
\Cref{alg:2/3-MMS} first invokes the \highValuedAlg (\Cref{alg:high-valued}).
According to \Cref{lem:n-agent-high-valued}, \Cref{alg:high-valued} outputs a $\frac{2}{3}$-MMS allocation for a subset of agents~$N' \subseteq N$ (implying $|N'| \in \{1, 2, 3\}$) and a reduced instance.
It now suffices to show the allocation for the remaining agents~$N \setminus N'$ is also $\frac{2}{3}$-MMS.
If $|N'| = 3$, our result immediately follows according to \Cref{lem:n-agent-high-valued}.
We next prove $|N'| \neq 2$ by contradiction.
Assume $|N'| = 2$ instead.
Since $N'' = N \setminus N'$, we thus have $|N''| = 3 - 2 = 1$.
Put differently, \Cref{alg:high-valued} terminates with one remaining agent who receives nothing during the execution of the algorithm.
This is, however, not possible because \Cref{alg:high-valued} hands the last agent the remaining goods, a contradiction.
Last, in the case that $|N'| = 1$, the current instance involves two agents who (i) value the current set of goods at least twice as much as their own MMS, and (ii) value each good less than $2/3$ times their own MMS.
Therefore, the $\twoAgentTwoThirdsAlloc$ (\Cref{alg:2-agent-2/3-MMS}) takes as input a valid instance and according to \Cref{thm:2-agent-2/3-MMS}, returns a $\frac{2}{3}$-MMS allocation for the remaining two agents, as desired.
\end{proof}

In the following, we assume there is no $\frac{2}{3}$-high-valued good in the instance.
We start with \crefrange{ALG-3:if-good-bundle}{ALG-3:good-bundle-call-2-agent-alg} of \Cref{alg:2/3-MMS}, where we deal with the \goodBundle if there exists one.

\begin{lemma}
\label{lem:3-agent-good-bundle}
Given a $3$-agent instance that does not contain any $\frac23$-high-valued good, if there exists a \goodBundle~$B$ with respect to some agent~$i \in N$, \Cref{alg:2/3-MMS} computes a $\frac{2}{3}$-MMS allocation.
\end{lemma}

\begin{proof}
By the definition of a \goodBundle in \Cref{def:reducible-bundle}, agent~$i$ gets utility at least~$2/3 \cdot \MMS_i$ by receiving bundle~$B$.
For ease of exposition, let $\{j, k\} = N \setminus \{i\} $ and assume without loss of generality that $\utility{j}{M \setminus B} \geq 2 \cdot \MMS_j$ and $\utility{k}{M \setminus B} \geq 4/3 \cdot \MMS_k$.
Then, due to \Cref{thm:2-agent-2/3-MMS}, \Cref{alg:2-agent-2/3-MMS} outputs a $\frac{2}{3}$-MMS allocation for agents~$j$ and~$k$, as desired.
\end{proof}

We next analyze the structure of instances that do not have \goodBundle{s}, and show that our algorithm (\crefrange{ALG-3:not-good-bundle}{ALG-3:common-div:last-agent}) still outputs a $\frac{2}{3}$-MMS allocation.

\begin{lemma}
\label{lem:alg-cover-instances}
Given a $3$-agent instance that does not contain any $\frac{2}{3}$-high-valued good, if there does not exist a \goodBundle, then the following conditions hold simultaneously:
\begin{enumerate}[label=(\roman*)]
\item\label{item:pair-of-goods}
for any pair of goods~$g, g' \in M$, either $\utility{a}{\{g, g'\}} > \MMS_a$ for all agents~$a \in N$ or $\utility{a}{\{g, g'\}} < 2/3 \cdot \MMS_a$ for all agents~$a \in N$;

\item\label{item:three-supporters}
there exists a pair of goods~$o, o' \in M$ such that for all agents~$a \in N$, $\utility{a}{\{o, o'\}} > \MMS_a$;

\item\label{item:valuable-small-goods}
for all agents~$a \in N$, $\sum_{g \in M \colon \utility{a}{g} \leq \frac{\MMS_a}{3}} \utility{a}{g} < \frac{\MMS_a}{6}$.
\end{enumerate}
\end{lemma}

\begin{proof}
We prove the statement by contrapositive.
Specifically, we show that if any of \cref{item:pair-of-goods,item:three-supporters,item:valuable-small-goods} is violated in the given instance, then there exits a \goodBundle.

If \cref{item:pair-of-goods} is violated, then there exists a pair of goods~$g, g' \in M$ such that for some agent~$i \in N$, $\utility{i}{\{g, g'\}} \geq 2/3 \cdot \MMS_i$ and for some other agent~$j \in N \setminus \{i\}$, $\utility{j}{\{g, g'\}} \leq \MMS_j$.
For the last agent~$k \in N \setminus \{i, j\}$, clearly, $\utility{k}{\{g, g'\}} < 2 \times 2/3 \cdot \MMS_k = 4/3 \cdot \MMS_k$.
As a result, $\{g, g'\}$ is a \goodBundle with respect to agent~$i$.

Given \cref{item:pair-of-goods}, if \cref{item:three-supporters} is violated, then for each agent~$a \in N$ and each pair of goods~$g, g' \in M$, we have $\utility{a}{\{g, g'\}} < 2/3 \cdot \MMS_a$.
Clearly, a \goodBundle, if exists, has size at least three.
Moreover, the third most valuable good of any agent is worth less than $1/3$ of her MMS.
For each~$a \in N$, let integral bundle~$B_a \subseteq M$ be of the minimal cardinal size such that $\utility{a}{B_a} \geq \MMS_a$.
Let~$i$ be the agent whose~$B_i$ is of the minimal cardinal size among the three agents (arbitrary tie-breaking).
Removing agent~$i$'s least valuable good from~$B_i$ gives us a \goodBundle~$B_i$ because $\utility{i}{B_i} \in [2/3 \cdot \MMS_i, \MMS_i]$ and $\utility{j}{M \setminus B_i} > 3 \cdot \MMS_j - \MMS_j = 2 \cdot \MMS_j$ for all~$j \in N \setminus \{i\}$.

Given \cref{item:pair-of-goods,item:three-supporters}, if \cref{item:valuable-small-goods} is violated, there exists an agent~$i \in N$ such that
\[
\textstyle
\sum_{g \in M \colon \utility{i}{g} \leq \frac{\MMS_i}{3}} \utility{i}{g} \geq \frac{\MMS_i}{6}.
\]
According to \cref{item:three-supporters}, we assume without loss of generality $\utility{i}{o} > \frac{\MMS_i}{2}$.
Let agent~$i$ add good~$o$ and next those goods of value at most~$\frac{\MMS_i}{3}$ to an empty bundle~$B$ until $\utility{i}{B} \in [2/3 \cdot \MMS_i, \MMS_i]$.
If there exists~$j \in N \setminus \{i\}$ such that $\utility{j}{B} \leq 5/3 \cdot \MMS_j$, $B$ is a \goodBundle with respect to the agent~$N \setminus \{i, j\}$ if she values~$B$ at least $2/3$ of her MMS or otherwise, agent~$i$.
Else, $\utility{a}{B} > 5/3 \cdot \MMS_a$ for all~$a \in N \setminus \{i\}$.
Fix any~$j \in N \setminus \{i\}$.
Let agent~$j$ remove goods from bundle~$B$ until $\utility{j}{B} \in [2/3 \cdot \MMS_j, 5/3 \cdot \MMS_j]$.
Now, $B$ is a \goodBundle with respect to the agent~$N \setminus \{i, j\}$ if she values~$B$ at least $2/3$ of her MMS and agent~$j$ otherwise.
\end{proof}

In addition to the structure of how agents value the goods, \Cref{alg:2/3-MMS} further examines agents' subjective divisibility.
Our next \namecref{lem:3-agent-handle-no-good-bundle} establishes the correctness of this part of our algorithm.

\begin{lemma}
\label{lem:3-agent-handle-no-good-bundle}
When the \cref{item:pair-of-goods,item:three-supporters,item:valuable-small-goods} in \Cref{lem:alg-cover-instances} hold, \Cref{alg:2/3-MMS} computes a $\frac{2}{3}$-MMS allocation.
\end{lemma}

\begin{proof}
We prove the statement by showing that \crefrange{ALG-3:not-good-bundle}{ALG-3:common-div:last-agent} output a $\frac{2}{3}$-MMS allocation.

According to \cref{item:valuable-small-goods}, for all~$a \in N$, we have $\sum_{g \in M \colon \utility{a}{g} \leq \frac{\MMS_a}{3}} \utility{a}{g} < \frac{\MMS_a}{6}$, meaning that for each~$g \in M$, either $\utility{a}{g} > \frac{\MMS_a}{3}$ or $\utility{a}{g} < \frac{\MMS_a}{6}$.
Consequently,
\[
\textstyle
\sum_{g \in M \colon \utility{a}{g} < \frac{\MMS_a}{6}} \utility{a}{g} < \frac{\MMS_a}{6}.
\]
Furthermore, $\utility{a}{M} \geq 3 \cdot \MMS_a$ implies that
\[
\textstyle
\sum_{g \in M \colon \utility{a}{g} > \frac{\MMS_a}{3}} \utility{a}{g}
= \utility{a}{M} - \sum_{g \in M \colon \utility{a}{g} < \frac{\MMS_a}{6}} \utility{a}{g}
> 3 \cdot \MMS_a - \frac{\MMS_a}{6}
> \frac{8}{3} \cdot \MMS_a.
\]
Recall that $\utility{a}{g} < 2/3 \cdot \MMS_a$ for all~$a \in N$ and~$g \in M$.
Therefore, there must be at least \emph{five} goods, each of which is worth greater than~$\frac{\MMS_a}{3}$.
We now show that each of those goods is worth greater than $1/3$ of any other agent's MMS as well.
Put differently, let
\[
\textstyle
\widehat{M} \coloneqq \{g \in M \mid \forall a \in N, \utility{a}{g} > \frac{\MMS_a}{3}\};
\]
we will show $|\widehat{M}| \geq 5$.
Fix any~$a \in N$.
Denote by~$\widehat{M}' \coloneqq \{g \in M \mid \utility{a}{g} > \frac{\MMS_a}{3}\}$.
From our previous argument, $|\widehat{M}'| \geq 5$.
For any pair of~$g, g' \in \widehat{M}'$, we have
\[
4/3 \cdot \MMS_a > \utility{a}{\{g, g'\}} > 2/3 \cdot \MMS_a.
\]
Since $\{g, g'\}$ is not a \goodBundle, it means that $\utility{a'}{\{g, g'\}} > \MMS_{a'}$ for all~$a' \in N$.
We thus conclude $\widehat{M}' \subseteq \widehat{M}$.
As a result, $|\widehat{M}| \geq |\widehat{M}'| \geq 5$.

Now, recall from \cref{item:three-supporters} that there exists a pair of~$o, o' \in M$ such that for all~$a \in N$, $\utility{a}{\{o, o'\}} > \MMS_a$.
We have $o, o' \in \widehat{M}$ by the definition of~$\widehat{M}$ stated in the last paragraph.
If an agent has at least four indivisible goods in~$\widehat{M}$, due to the pigeonhole principle, at least two of them (say $g, g' \in \widehat{M}$) must be included in the same bundle in her MMS partition.
Then, $\{g, g'\}$ is a \goodBundle so the case would have been already solved in \crefrange{ALG-3:if-good-bundle}{ALG-3:good-bundle-call-2-agent-alg}.
Thus, each agent has at most three indivisible goods in~$\widehat{M}$.
Alternatively, it means that there exists some good that is considered as divisible for \emph{at least two agents}.
For simplicity, we refer to the good as a \emph{common} divisible good.
Assume without loss of generality that both agents~$i, j \in N$ regard good~$d \in \widehat{M}$ as divisible.
Our main idea here is to partition \emph{three} goods between the two agents~$i$ and~$j$ so that they are both satisfied.
To be more specific, in addition to good~$d$, we pick the goods~$o, o'$.
In case that $d \in \{o, o'\}$, we pick arbitrarily another good from~$\widehat{M} \setminus \{o, o'\}$.
Let us refer to the two picked goods other than~$d$ as~$g_1, g_2$, respectively.
Clearly,
\begin{equation}
\label{eq:common-valuable-div}
\textstyle
\utility{i}{\{g_1, d, g_2\}} > \MMS_i + \frac{\MMS_i}{3} = 4/3 \cdot \MMS_i;
\end{equation}
a similar inequality also holds for~$j$.
With this relationship in hand, intuitively, by dividing good~$d$, we are able to find an allocation of goods~$g_1, d, g_2$ to~$i, j$ such that both agents get a utility of at least~$2/3$ of their own MMS.
Arrange the three goods~$g_1, d, g_2$ on a line in this order.
Next, let agent~$i$ cut the line into two bundles~$B, B'$ of equal value, i.e., $\utility{i}{B} = \utility{i}{B'}$.
Since $\utility{i}{g_1}, \utility{i}{g_2} < 2/3 \cdot \MMS_i$, the cut must be on the common divisible good~$d$.
It means that agent~$i$ indeed has the $2$-partition~$(B, B')$ of goods~$g_1, d, g_2$ such that
\[
\utility{i}{B} = \utility{i}{B'} > 2/3 \cdot \MMS_i,
\]
where the last inequality follows from \Cref{eq:common-valuable-div}.
Agent~$j$'s preferred bundle between~$B$ and~$B'$ is clearly worth at least~$2/3 \cdot \MMS_j$.
In short, both agents~$i$ and~$j$ are satisfied.
As the third agent~$N \setminus \{i, j\}$ values~$\{g_1, d, g_2\}$ at most twice as much as her MMS, the remaining goods~$M \setminus \{g_1, d, g_2\}$ is worth at least her MMS, as desired.
\end{proof}

\bigskip

The correctness of \Cref{thm:3-agent-2/3-MMS} follows from \Cref{lem:3-agent-2/3-MMS:high,lem:3-agent-good-bundle,lem:alg-cover-instances,lem:3-agent-handle-no-good-bundle} stated above.

\subsection{Any Number of Agents}
\label{sec:MMS:n-agents}

We start by showing that the MMS approximation guarantee is non-increasing with respect to~$n$.
Given an instance, let its \emph{MMS approximation guarantee} be the maximum value of~$\alpha$ such that the instance admits an $\alpha$-MMS allocation.

\begin{theorem}
\label{thm:impossibility-non-increasing}
Given an instance~$I$ with~$n$ agents and~$m$ goods, whose MMS approximation guarantee is~$\gamma$, there exists an instance~$I'$ with~$n+1$ agents and~$m+1$ goods such that the MMS approximation guarantee of instance~$I'$ is at most~$\gamma$.
\end{theorem}

\begin{proof}
\begin{table}[t]
\centering
\begin{tabular}{|*{3}{c|}|c|}
\hline
Agents & $M$ (the set of $m$ goods) & $g_{m+1}$ & Agents' Maximin Share \\
\hline
$a_1$ & & ind, $\MMS_1(n, M)$ & $\MMS_1(n+1, M') = \MMS_1(n, M)$ \\
\cline{1-1}
\cline{3-4}
\vdots & subjective divisibility & \vdots & \vdots \\
\cline{1-1}
\cline{3-4}
$a_i$ & and utilities & ind, $\MMS_i(n, M)$ & $\MMS_i(n+1, M') = \MMS_i(n, M)$ \\
\cline{1-1}
\cline{3-4}
\vdots & & \vdots & \vdots \\
\cline{1-1}
\cline{3-4}
$a_n$ & & ind, $\MMS_n(n, M)$ & $\MMS_n(n+1, M') = \MMS_n(n, M)$ \\
\hline
$a_{n+1}$ & indivisible and $0$, $\forall g \in M$ & div, $1$ & $1 / (n+1)$ \\
\hline
\end{tabular}
\caption{Instance~$I'$ constructed in the proof of \Cref{thm:impossibility-non-increasing}.
Let $M' \coloneqq M \cup \{g_{m+1}\}$.}
\label{table:n-agent-construction-non-increasing}
\end{table}

Following the theorem statement, we have an instance~$I$ with~$n$ agents and~$m$ goods, whose MMS approximation guarantee is~$\gamma$.
In what follows, we will construct an instance~$I'$ as shown in \Cref{table:n-agent-construction-non-increasing}
with $n+1$ agents and $m+1$ goods such that its MMS approximation guarantee is at most~$\gamma$.
Instance~$I'$ consists of instance~$I$, agent~$n+1$ and good~$g_{m+1}$.
Let $M' \coloneqq M \cup \{g_{m+1}\}$ denote the goods in Instance~$I'$.
For each agent~$i \in [n]$, agent~$i$ regards good~$g_{m+1}$ as indivisible and values the good exactly at $\MMS_i(n, M)$.
Agent~$n+1$ regards every good in~$M$ as indivisible and values all goods at zero; the agent regards good~$g_{m+1}$ as divisible and has value~$1$ for the good.

As for agents' maximin share in instance~$I'$, clearly, we have $\MMS_{n+1}(n+1, M') = 1 / (n+1)$.
Next, we show that for each agent~$i \in [n]$, we have $\MMS_i(n+1, M') = \MMS_i(n, M)$, i.e., agent~$i$'s maximin share in instance~$I'$ is exactly the same as that in instance~$I$.
First, we have $\MMS_i(n+1, M') \geq \MMS_i(n, M)$ in that the MMS partition of agent~$i$ in instance~$I$ plus a singleton bundle containing good~$g_{m+1}$ is a valid partition in instance~$I'$.
Second, we show that $\MMS_i(n+1, M') \leq \MMS_i(n, M)$.
Suppose for the sake of contradiction that $\MMS_i(n+1, M') > \MMS_i(n, M)$.
We note that since good~$g_{m+1}$ is indivisible for agent~$i$, this good is included entirely in a single bundle of a MMS partition in instance~$I'$; removing this bundle gives a better MMS partition in instance~$I$, a contradiction.

Now, we focus on the MMS approximation guarantee for instance~$I'$.
First, to have a positive MMS approximation guarantee, agent~$n+1$ should get a positive amount of good~$g_{m+1}$.
Once agent~$n+1$ gets some part of good~$g_{m+1}$, other agents value the remaining part of good~$g_{m+1}$ at zero, so there is no point to give the remainder of good~$g_{m+1}$ to any of them.
Then, we discuss how goods~$M$ are allocated.
There is no point to give these goods to agent~$n+1$ because this agent does not value these goods at all.
Since now goods~$M$ are allocated between agents~$[n]$, according to the theorem statement that the MMS approximation guarantee for instance~$I$ is~$\gamma$, we conclude that the MMS approximation guarantee for instance~$I'$ is at most~$\gamma$.
\end{proof}

It is of interest to know if the worst-case MMS approximation guarantee would \emph{strictly} decrease at some point.
Any improvement of this direction strengthens the impossibility result below.

\begin{corollary}
\label{coro:n-agent-impossibility}
The worst-case MMS approximation guarantee is at most~$2/3$ for $n \geq 2$ agents.
\end{corollary}

On the positive side, we show below a $\frac{1}{2}$-MMS allocation always exists and can be found by \Cref{alg:n-agent-1/2-MMS}.
To achieve this, we first overload our notion of \emph{reducible-bundle} for the $\frac{1}{2}$-MMS target.

\begin{definition}[$\frac{1}{2}$-reducible bundle]
\label{def:1/2-reducible-bundle}
An integral bundle~$B$ is a \emph{$\frac{1}{2}$-reducible bundle with respect to agent~$i \in N$} if $u_i(B) \geq \MMS_i / 2$ and $u_j(M \setminus B) \geq (|N| - 1) \cdot \MMS_j$ for all other agents~$j \in N \setminus \{i\}$.
\end{definition}

\begin{algorithm}[t]
\caption{$n$-agent-$\frac{1}{2}$-MMS Algorithm}
\label{alg:n-agent-1/2-MMS}
\DontPrintSemicolon

\KwIn{Agents~$N = [n]$, goods~$M = \{g_1, g_2, \dots, g_m\}$, subjective divisibility, and utilities.}
\KwOut{A $\frac{1}{2}$-MMS allocation for agents~$N$.}

\lForEach{$i \in N$}{
	Initialize $A_i \gets \emptyset$ and compute the agent's maximin share~$\MMS_i(n, M)$.
}

$(N', (A_i)_{i \in N'}, N'', M'') \gets \highValuedAlg(1/2, N, M, (\MMS_i)_{i \in N})$ \tcp*{\Cref{alg:high-valued}.}

\While {$|N''| > 0$} {
	Find a $\frac{1}{2}$-reducible bundle $B \subseteq M''$ with respect to some agent $i \in N''$.\; \label{ALG-n:find-1/2-reducible-bundle}
	$A_i \gets B$, $M'' \gets M'' \setminus B$, $N'' \gets N'' \setminus \{i\}$\;
}

\Return{$(A_1, A_2, \dots, A_n)$}
\end{algorithm}

\Cref{alg:n-agent-1/2-MMS} for finding a $\frac{1}{2}$-MMS allocation for $n$ agents is much simpler: we first invoke \highValuedAlg (\Cref{alg:high-valued}) with $\beta = 1/2$ to assign all $\frac{1}{2}$-high-valued goods to a subset of agents.
Then, in the smaller instance without any $\frac{1}{2}$-high-valued goods, we can show that a $\frac{1}{2}$-reducible bundle always exists, allowing us to iteratively assigning bundles to all remaining agents.

\begin{theorem}
\label{thm:n-agent-1/2-MMS}
\Cref{alg:n-agent-1/2-MMS} computes a $\frac{1}{2}$-MMS allocation for $n$ agents.
\end{theorem}

\begin{proof}
Given any instance, \Cref{alg:n-agent-1/2-MMS} first invokes \Cref{alg:high-valued} to handle allocations involving $\frac{1}{2}$-high-valued goods.
According to \Cref{lem:n-agent-high-valued}, \Cref{alg:high-valued} outputs a $\frac{1}{2}$-MMS allocation for agent~$N' \subseteq N$ and a (reduced) instance with agents~$N'' = N \setminus N'$ and goods~$M'' = M \setminus \bigcup_{i \in N'} A_i$ in the sense that for each~$i \in N''$, $\utility{i}{M''} \geq |N''| \cdot \MMS_i$ and for each~$g \in M''$, $\utility{i}{g} < \MMS_i / 2$.
Since for all~$j \in N''$ and~$g \in M''$, $\utility{j}{g} < \MMS_j / 2$, in each iteration of the \verb|while|-loop, a $\frac{1}{2}$-\goodBundle can be found in \cref{ALG-n:find-1/2-reducible-bundle} by adding one good at a time from~$M''$ to an empty bundle~$B$ until $\utility{i}{B} \in [\MMS_i / 2, \MMS_i]$ for some~$i \in N''$.
Clearly, $\utility{j}{B} \leq \MMS_j$.
It means that in each iteration, agent~$i$ is satisfied by receiving bundle~$B$ and the remaining agents have enough value for the remaining goods.
\end{proof}

\medskip

Closing the gap between \Cref{coro:n-agent-impossibility} and \Cref{thm:n-agent-1/2-MMS} for~$n \geq 4$ agents remains an intriguing open question.
We conjecture that $\frac{2}{3}$-MMS is the ultimate answer.

\subsection{Computation and Economic Efficiency}
\label{sec:MMS:discussion}

We first discuss the computational issues revolved around our algorithms and next move on to the additional consideration of achieving economic efficiency.

\subsubsection*{Computation}

\Cref{alg:high-valued,alg:2-agent-2/3-MMS,alg:2/3-MMS,alg:n-agent-1/2-MMS} do not run in polynomial time as they need agents' maximin share.
Nevertheless, by using the PTAS of \citet[Lemma~4]{BeiLiLu21} to compute agents' \emph{approximate} maximin share, we can slightly modify \Cref{alg:2-agent-2/3-MMS,alg:2/3-MMS} and compute a $\left( (1 - \varepsilon) \cdot \frac{2}{3} \right)$-MMS allocation in the $2$- and $3$-agent cases, respectively, in polynomial time:
\begin{itemize}
\item When we need to compute an agent's MMS, we instead compute its $(1 - \varepsilon)$-approximation via the PTAS of \citet{BeiLiLu21}.

\item We then replace all exact MMS values used in our algorithms with their $(1 - \varepsilon)$-approximations.
Throughout the process of our algorithms, only elementary operations are applied to the MMS values, and the multiplicative $(1 - \varepsilon)$ factor can be readily pulled out to the front, which eventually lead to a $(1 - \varepsilon)$-approximate $\frac{2}{3}$-MMS allocation.
\end{itemize}

Other steps in \Cref{alg:2-agent-2/3-MMS} not involving MMS values can be easily implemented efficiently.

For \Cref{alg:2/3-MMS}, as it is not trivial to decide the existence of a \goodBundle (and compute one when it exists), we describe a computationally efficient method to implement the algorithm.
To be more specific, we use conditions~\ref{item:pair-of-goods}, \ref{item:three-supporters} and~\ref{item:valuable-small-goods} of \Cref{lem:alg-cover-instances} to compute a \goodBundle.
Checking these conditions can be done in polynomial time by simple enumeration.
If any condition fails to hold, we will be able to find a reducible bundle efficiently by \Cref{lem:alg-cover-instances}.
When all three conditions hold simultaneously, although this does not necessarily rule out the existence of a \goodBundle, \Cref{alg:2/3-MMS} will proceed to \crefrange{ALG-3:not-good-bundle}{ALG-3:common-div:last-agent}, and the correctness follows from \Cref{lem:3-agent-handle-no-good-bundle}.

Analogously, by using $(1 - \varepsilon)$-approximate MMS values in \Cref{alg:n-agent-1/2-MMS}, we can obtain a $\left( (1 - \varepsilon) \cdot \frac{1}{2} \right)$-MMS for $n$ agents in polynomial time, following similar arguments as in the two- and three-agent cases.

\subsubsection*{Economic Efficiency}

Existence-wise, whenever an $\alpha$-MMS allocation is guaranteed to exist, an $\alpha$-MMS and PO allocation always exists.
This is simply because any Pareto improvement of the $\alpha$-MMS allocation remains to be $\alpha$-MMS and the $\alpha$-MMS allocation which does not admit any Pareto improvement is Pareto optimal by definition.
We thus have the following corollaries:

\begin{corollary}
In the $2$- and $3$-agent cases, a $\frac{2}{3}$-MMS and PO allocation always exists.

In the $n$-agent case, a $\frac{1}{2}$-MMS and PO allocation always exists.
\end{corollary}

Despite the existence argument, it is, however, NP-hard to find Pareto improvements even for indivisible-goods allocation~\citep[see, e.g.,][]{AzizBiLa16,deKeijzerBoKl09}.
It is an interesting direction to design computationally efficient algorithms which outputs allocations that are both PO and have good approximation to MMS.

\section{Envy-freeness Relaxations}
\label{sec:EFM}

This section is concerned with envy-freeness relaxations, which are comparison-based fairness notions.
We start by introducing relevant fairness concepts.
An allocation~$\alloc = (A_1, A_2, \dots, A_n)$ is said to satisfy \emph{envy-freeness (\EF{})} if for any pair of agents~$i, j \in N$, $u_i(A_i) \geq u_i(A_j)$~\citep{Foley67}.
Relaxations of envy-freeness have been proposed for indivisible-goods allocation~\citep{Budish11,CaragiannisKuMo19,LiptonMaMo04}.
An allocation~$\alloc = (A_1, A_2, \dots, A_n)$ of indivisible goods is said to satisfy
\begin{itemize}
\item \emph{envy-freeness up to \emph{any} good (\EFX)} if for any pair of agents~$i, j \in N$ with $A_j \neq \emptyset$ and any good~$g \in A_j$ such that $u_i(g) > 0$, we have $u_i(A_i) \geq u_i(A_j \setminus \{g\})$;

\item \emph{envy-freeness up to \emph{one} good (\EFOne)} if for any pair of agents~$i, j \in N$ and~$A_j \neq \emptyset$, there exists some good~$g \in A_j$ such that $u_i(A_i) \geq u_i(A_j \setminus \{g\})$.
\end{itemize}

Clearly, \EF{} $\implies$ \EFX $\implies$ \EFOne.
In the mixed-goods model where all agents agree on whether a good is divisible or indivisible, \citet{BeiLiLi21} introduced \emph{envy-freeness for mixed goods (\EFM)}, a natural generalization of both \EF{} and \EFOne.
The notion can be further strengthened by replacing \EFOne with \EFX~\citep{BeiLiLi21,NishimuraSu23}.
We adapt the notions to our setting as follows.

\begin{definition}[\EFM and \EFXM]
\label{def:EFM}
An allocation~$(A_1, A_2, \dots, A_n)$ is said to satisfy
\begin{itemize}
\item \emph{\EFM} if for any pair of agents~$i, j \in N$,
\begin{itemize}
\item if $g \in \indGoodsOf{i}$ for all~$g \in A_j$, there exists some~$g \in A_j$ such that $u_i(A_i) \geq u_i(A_j \setminus \{g\})$;
\item otherwise, $u_i(A_i) \geq u_i(A_j)$.
\end{itemize}

\item \emph{\EFXM} if for any pair of agents~$i, j \in N$,
\begin{itemize}
\item if $g \in \indGoodsOf{i}$ for all~$g \in A_j$, then for any~$g \in A_j$ with $u_i(g) > 0$, $u_i(A_i) \geq u_i(A_j \setminus \{g\})$;
\item otherwise, $u_i(A_i) \geq u_i(A_j)$.
\end{itemize}
\end{itemize}
\end{definition}

Put differently, \EFM requires that either agent~$i$ is envy-free towards agent~$j$, or agent~$i$ regards all goods in~$A_j$ as indivisible and the envy from agent~$i$ to agent~$j$ can be eliminated by (hypothetically) removing some good from~$A_j$.
In a similar vein, the stronger \EFXM requires that if agent~$i$ envies agent~$j$, then agent~$i$ regards all goods in~$A_j$ as indivisible and the envy can be eliminated by (hypothetically) removing any positively valued good from~$A_j$.

Another relaxation that has been investigated in the mixed-goods model is \emph{envy-freeness up to one good for mixed goods (\EFoneM)}, a weaker notion than \EFM~\citep{CaragiannisKuMo19,NishimuraSu23}.
The concept can be adapted to our setting in the following natural way.

\begin{definition}[\EFoneM]
An allocation $(A_1, A_2, \dots, A_n)$ is said to satisfy \emph{\EFoneM} if for any pair of agents~$i, j \in N$,
\begin{itemize}
\item if $A_j$ contains some good $g$ that agent $i$ considers indivisible and values positively (i.e., $g \in M^{IND}_i$ and $u_i(g) > 0$), then $u_i(A_i) \geq u_i(A_j \setminus \{g\})$;

\item otherwise, $u_i(A_i) \geq u_i(A_j)$.
\end{itemize}
\end{definition}

In words, as long as there exists agent~$i$'s indivisible good being included integrally in~$A_j$, \EFoneM requires that agent~$i$ should be envy-free towards agent~$j$ once agent~$i$'s most valued indivisible good is removed from~$A_j$; otherwise, agent~$i$ does not envy agent~$j$.
Clearly, \EF{} $\implies$ \EFXM $\implies$ \EFM $\implies$ \EFoneM.

As mentioned earlier in \Cref{sec:introduction}, envy-freeness (and hence \EFXM, \EFM and \EFoneM) alone can be satisfied vacuously by partitioning each good into $n$ equal parts and giving each agent one part.
The allocation is apparently undesirable, especially when all goods are indivisible for all agents.
We therefore associate (almost) envy-free allocations with economic efficiency considerations, and additionally define the following minimal property of efficiency, a much weaker notion than PO.

\begin{definition}[Non-wastefulness]
\label{def:non-wastefulness}
An allocation $(A_1, A_2, \dots, A_n)$ is said to be \emph{non-wasteful} if for any good~$g \in M$ and~$i \in N$ such that $g^{A_i} \neq \emptyset$, we have $\utility{i}{g^{A_i}} > 0$.
\end{definition}

In a non-wasteful allocation, any good in an agent's bundle yields positive utility.
Non-wastefulness excludes the scenario where agents receive a fraction of their indivisible goods.

\subsection{EFM and EFXM}

We first study to what extent \EFM and \EFXM is compatible with economic efficiency notions.
The vacuous \EFM\ / \EFXM allocation mentioned above violates non-wastefulness.
As a matter of fact, we show below that \EFM is incompatible with non-wastefulness, even for two agents.

\begin{theorem}
\label{thm:EFM-non-wasteful-incompatible}
For any set of agents $N = [n]$ with $n \geq 2$, there exists a set of goods $M$ of size $m = n+1$ that do not admit an \EFM (hence, \EFXM) and non-wasteful allocation.
\end{theorem}

\begin{proof}[Proof (for two agents)]
It suffices to show \EFM is incompatible with non-wastefulness.
We provide below an example for two agents, and defer the $n$-agent case in \Cref{app:omitted-proofs}.
Consider the following instance with $0 < \varepsilon < 1/2$:
\begin{center}
\begin{tabular}{@{}c|*{3}{l}@{}}
\toprule
& $g_0$ & $g_1$ & $g_2$ \\
\midrule
Agent~$1$ & ind, $1 - \varepsilon / 2$ & div, $\varepsilon$ & ind, $1 - \varepsilon$ \\
Agent~$2$ & ind, $1 - \varepsilon / 2$ & ind, $1 - \varepsilon$ & div, $\varepsilon$ \\
\bottomrule
\end{tabular}
\end{center}

First of all, in any complete allocation, non-wastefulness requires all goods to be allocated in their entirety.
Second, in any EFM allocation, no agent gets an empty bundle because they all regard some good as divisible and would not be envy-free towards the other agent by being empty-handed.
Hence, one agent (denoting as agent~$a$) receives one good and the other agent (denoting as agent~$b$) receives two goods.
\begin{itemize}
\item If agent~$a$ gets her unique divisible good, she still envies agent~$b$ even if some good is removed from agent~$b$'s bundle.
\item If agent~$a$ gets one of her indivisible good, she envies agent~$b$ even though agent~$b$'s bundle contains agent~$a$'s divisible good.
\end{itemize}
In either case, \EFM is violated.
The conclusion follows.
\end{proof}

\Cref{thm:EFM-non-wasteful-incompatible} renders a strong conflict between economic efficiency and fairness when allocations are required to be complete.
If \emph{partial} allocations are allowed, that is, only a subset of the goods~$M$ is allocated, our next result shows for two agents, an \EFXM (and therefore \EFM) and non-wasteful allocation is always attainable when \emph{at most} one good is discarded.
Due to our subjective divisibility assumption, we allow some good to be partially discarded.

This result is inspired by the line of work which, given that a complete \EFX allocation is only known to exist for at most three agents~\citep{ChaudhuryGaMe24,PlautRo20}, proves the existence of a \emph{partial} \EFX allocation of indivisible goods, subject to few unallocated goods (i.e., charity)~\citep{AkramiAlCh23,BergerCoFe22,ChaudhuryKaMe21}.
It is worth noting that in our context, this approach helps us circumvent an \emph{impossibility}, and ensure the existence of \EFXM and non-wasteful allocations for two agents by discarding at most one good.

\begin{algorithm}[t]
\caption{$2$-Agent \EFXM and Non-wasteful Allocation with Charity}
\label{alg:2-agent-charity}
\DontPrintSemicolon

\KwIn{Two agents~$N = [2]$, goods~$M$, subjective divisibility, and utilities.}
\KwOut{An \EFXM and non-wasteful allocation with at most one good being unallocated.}

$A_1, A_2 \gets \emptyset$\;
Temporarily regard all goods as indivisible. Feed two copies of agent~$1$ into Algorithm~6.1 of \citet{PlautRo20} to have an integral \EFX partition~$(P_1, P_2)$ for agent~$1$. Assume w.l.o.g.\ that agent~$1$'s zero-valued goods are in~$P_2$, and $\utility{1}{P_1} \geq \utility{1}{P_2}$.\; \label{ALG-2:charity:a1-two-copies}

\If{$u_1(P_1) = u_1(P_2)$ or $u_2(P_1) \leq u_2(P_2)$}{
	Give agent~$2$ her preferred bundle between~$P_1$ and~$P_2$ and agent~$1$ the other.\;
	Both agents move their zero-valued goods from their own bundle to the other agent's.\;
	\Return{$(A_1, A_2)$}
}

\tcp{$u_1(P_1) > u_2(P_2)$ and $u_2(P_1) > u_2(P_2)$.}
\lIf{$\divGoodsOf{1} \cap P_1 = \emptyset$}{ \label{ALG-2:charity:all-indivisible}
	$(A_1, A_2) \gets (P_2, P_1)$ \label{ALG-2:charity:all-indivisible-alloc}
}\Else(\tcp*[h]{Agent~$1$ regards some good in~$P_1$ as divisible.}){
	\eIf(\tcp*[h]{Agents~$1, 2$ have common divisible goods in~$P_1$.}){$\divGoodsOf{1} \cap \divGoodsOf{2} \cap P_1 \neq \emptyset$}{
		Pick any~$o \in \divGoodsOf{1} \cap \divGoodsOf{2} \cap P_1$, and divide~$o$ into two pieces~$o'$ and $o''$ such that $u_1(P'_1 \coloneqq P_1 \setminus \{o\} \cup \{o'\}) = u_1(P'_2 \coloneqq P_2 \cup \{o''\})$.\;
		Give agent~$2$ her preferred bundle between~$P'_1$ and~$P'_2$ and agent~$1$ the other.\;
	}(\tcp*[h]{Agents~$1$ and~$2$ do not have common divisible goods in~$P_1$.}){
		Pick any~$o \in \divGoodsOf{1} \cap P_1$.\;
		\lIf{$u_2(P_1 \setminus \{o\}) \geq u_2(P_2)$}{
			$(A_1, A_2) \gets (P_2, P_1 \setminus \{o\})$
		}\Else{
			Let~$\widetilde{o} \subsetneq o$ such that $u_1(P_1 \setminus \{o\} \cup \{\widetilde{o}\}) \geq u_1(P_2)$.\;
			$(A_1, A_2) \gets (P_1 \setminus \{o\} \cup \{\widetilde{o}\}, P_2)$\;
		}
	}
}

Let both agents move their zero-valued goods from their own bundle to the other agent's.\;

\Return{$(A_1, A_2)$}
\end{algorithm}

\begin{theorem}
\label{thm:2-agent-EFXM-non-wasteful-w.-charity}
Given any $2$-agent instance with goods~$M$, \Cref{alg:2-agent-charity} computes an \EFXM and non-wasteful allocation in polynomial time if at most one good from~$M$ is unallocated.
\end{theorem}

\begin{proof}
We start by showing the allocation returned by \Cref{alg:2-agent-charity} is \EFXM and non-wasteful.
First of all, we temporarily regard all goods as indivisible, and obtain an integral partition~$(P_1, P_2)$ of goods~$M$ that is \EFX for agent~$1$ by executing Algorithm~6.1 of \citet{PlautRo20} on two copies of agent~$1$.
Assume without loss of generality that agent~$1$'s zero-valued goods are in~$P_2$, and $u_1(P_1) \geq u_1(P_2)$.

If $u_1(P_1) = u_1(P_2)$ or $u_2(P_1) \leq u_2(P_2)$, we give agent~$2$ her preferred bundle between~$P_1$ and~$P_2$ and agent~$1$ the other.
This allocation is \EF{}.
In order to fulfil non-wastefulness, let both agents move their zero-valued goods from their own bundle to the other agent's bundle.
The updated allocation is still \EF{}.
We thus have a complete allocation that is envy-free and non-wasteful.

In the following, we focus on the case where $u_1(P_1) > u_1(P_2)$ and $u_2(P_1) > u_2(P_2)$.
Recall that due to how we get partition~$(P_1, P_2)$, we have $u_1(P_1 \setminus \{g\}) \leq u_1(P_2)$ for all~$g \in P_1$.
If agent~$1$ regards all goods in~$P_1$ as indivisible, then following the definition of \EFXM (\Cref{def:EFM}), allocation~$(P_2, P_1)$ satisfies \EFXM.
This allocation fails non-wastefulness because~$P_2$ contains agent~$1$'s zero-valued goods.
Let both agents exchange their own zero-valued good, and the resulting allocation is \EFXM and non-wasteful.
We then move on to the scenario where agent~$1$ regards some good in~$P_1$ as divisible, and distinguish the following two cases based on whether agents~$1$ and~$2$ have a common divisible good in~$P_1$ or not.
\begin{itemize}
\item $\divGoodsOf{1} \cap \divGoodsOf{2} \cap P_1 \neq \emptyset$.
We pick any common divisible good~$o \in \divGoodsOf{1} \cap \divGoodsOf{2} \cap P_1$, and divide~$o$ into two pieces~$o'$ and~$o''$ such that agent~$1$ values $P_1' \coloneqq P_1 \setminus \{o\} \cup \{o'\}$ and $P'_2 \coloneqq P_2 \cup \{o''\}$ equally.
Next, we give agent~$2$ her preferred bundle between~$P_1'$ and~$P'_2$ and agent~$1$ the other.
The resulting allocation is \EF{}.

\item $\divGoodsOf{1} \cap \divGoodsOf{2} \cap P_1 = \emptyset$.
Consider any~$o \in \divGoodsOf{1} \cap P_1$.
If $u_2(P_1 \setminus \{o\}) \geq u_2(P_2)$, allocation $(P_2, P_1 \setminus \{o\})$ is \EF{}.
Otherwise, we have $u_2(P_1 \setminus \{o\}) < u_2(P_2)$.
Let $\widetilde{o} \subsetneq o$ be such that $u_1(P_1 \setminus \{o\} \cup \{\widetilde{o}\}) \geq u_1(P_2)$.
Such an operation is valid because $u_1(P_1) > u_1(P_2)$, $u_1(P_1 \setminus \{o\}) \leq u_1(P_2)$, and good~$o$ is divisible from agent~$1$'s perspective.
Note also that since~$o$ is indivisible from agent~$2$'s perspective, $u_2(P_1 \setminus \{o\} \cup \{\widetilde{o}\}) = u_2(P_1 \setminus \{o\}) < u_2(P_2)$.
Therefore, allocation~$(P_1 \setminus \{o\} \cup \{\widetilde{o}\}, P_2)$ is \EF{}.
\end{itemize}
Again, by exchanging the agents' zero-valued goods, the updated allocations remain \EF{} and, moreover, satisfy non-wastefulness.
It is clear that at most one good is unallocated.

Finally, we analyze the running time of \Cref{alg:2-agent-charity}.
Due to \citet[Theorem~6.2]{PlautRo20}, \cref{ALG-2:charity:a1-two-copies} terminates in polynomial time.
All remaining steps are simple and can be implemented in polynomial time as well.
\end{proof}

If considering the stronger \EF{} notion, the following example demonstrates that even if we allow to dispose of at most one good, an \EF{} allocation is not guaranteed to exist.

\begin{example}
Consider an instance with two agents, two (objective) indivisible goods and the agents have an identical valuation~$u$ such that $u(g_1) = 2$ and $u(g_2) = 1$.
It can be verified that
\begin{itemize}
\item if we allocate both goods to agents~$1$ and~$2$, there is no envy-free allocation;
\item if we dispose of one good, again, there is no envy-free allocation by allocating the remaining single indivisible good between two agents.
\end{itemize}
\end{example}

It is an interesting open question of whether \Cref{thm:2-agent-EFXM-non-wasteful-w.-charity} can be further extended by showing that an almost-complete \EFXM (or even \EFM) and non-wasteful allocation always exists for any number of agents.
We briefly discuss the reason why it becomes more complicated when considering three or more agents.
On the one hand, the existing algorithmic idea to achieve \EFM heavily relies on the following steps~\citep[see, e.g.,][]{BeiLiLi21,BhaskarSrVa21}:
\begin{enumerate}[label=(\roman*)]
\item starting with an \EFOne allocation of indivisible items,
\item iteratively allocating some divisible resources to a set of identified agents in a ``perfect'' way, that is, all agents values all pieces equally, and
\item rotating bundles if necessary.
\end{enumerate}
The technique to achieve \EFM naturally violates non-wastefulness whenever we allocate divisible goods or rotate bundles, as a result, even if we discard some goods, it is still challenging to simultaneously achieve \EFM and non-wastefulness.
On the other hand, the proof of \Cref{thm:2-agent-EFXM-non-wasteful-w.-charity} hinges on the $2$-agent \EFX algorithm.
As mentioned earlier, the existence of \EFX allocations is largely open.
Moreover, although we have $3$-agent \EFX algorithms~\citep{ChaudhuryGaMe24,AkramiAlCh23}, the resulting \EFX allocation is based on careful case analysis.
It is unclear if there exists a systematic way to achieve \EFM and non-wastefulness with the possibility of discarding some goods, even for the $3$-agent case.

\subsection{EF1M: A Compatible Relaxation}

Having established the incompatibility of \EFM with non-wastefulness, we now turn to a weaker notion -- \EFoneM.
Our main result in this section shows \EFoneM is compatible with non-wastefulness.
Moreover, such an allocation can be computed efficiently.

\begin{theorem}
\label{thm:EF1M_nonwasteful}
For any number of agents, an \EFoneM and non-wasteful allocation always exists and can be found in polynomial time.
\end{theorem}

To prove Theorem 4.7, we begin by classifying the goods $M$ into two disjoint sets based on how agents view their divisibility:
\begin{itemize}
\item Let $S_1$ be the set of goods where at most one agent considers each good to be divisible.
\item Let $S_2$ be the set of goods where at least two agents consider each good to be divisible.
\end{itemize}

Observe that if we can find an EF1M and non-wasteful allocation of goods~\( S_1 \), then it is straightforward to extend this allocation by adding good~\( S_2 \) while still preserving these properties, because, since goods in \( S_2 \) are considered divisible by at least two agents, any fractional allocation of these goods can be adjusted in a way that maintains EF1M without violating non-wastefulness.
Thus, our primary challenge is to ensure an EF1M and non-wasteful allocation of good~\( S_1 \).

\begin{lemma}
\label{lemma:EF1M_nonwasteful_S1}
There exists an allocation of the goods in \( S_1 \) that satisfies both EF1M and non-wastefulness. Moreover, such an allocation can be found in polynomial time.
\end{lemma}

The proof of Lemma~\ref{lemma:EF1M_nonwasteful_S1} relies on a generalized version of the round-robin algorithm, which is traditionally used to produce EF1 allocations of indivisible goods.
The algorithm proceeds in rounds, and in each round~$k$, we select an arbitrary permutation of the agents, denoted as~$\pi_k$.
Following the order specified by~$\pi_k$, each agent selects their most preferred good from the set of remaining available goods in that round.
It is known that this round-robin algorithm always produces an EF1 allocation for indivisible goods~\citep[see, e.g.,][Section~3.1]{AmanatidisAzBi23}.
In this proof, we extend the use of this algorithm to handle the goods in~$S_1$ by leveraging the fact that each good in~$S_1$ is considered divisible by at most one agent.

First, we introduce the following notations for the generalized round-robin algorithm applied to the goods in~$S_1$.
In round~$k$, with a slight abuse of notation, let~$R$ denote the remaining set of goods in~$S_1$ that have not yet been allocated to any agent.
For each good~$g \in R$, we define $d(g)$ to be the agent who considers~$g$ to be divisible, and let $d(g) = \emptyset$ if no agent considers~$g$ to be divisible.
Note that $d(g)$ is well-defined due to the property of the set~$S_1$.

Next, we describe the construction of the permutation~$\pi_k$ of agents for round~$k$.
We define a directed graph~$G_k$ as follows:
\begin{itemize}
\item Each vertex represents an agent.
\item For each agent~$i$, let $g^i \in R$ be their most valuable remaining good, breaking ties arbitrarily.
\item If $d(g^i) \neq \emptyset$, add a directed edge from~$i$ to~$d(g^i)$
\end{itemize}
Note that each vertex has outdegree of at most~$1$ in~$G_k$.

\begin{algorithm}[t]
\caption{Generalized Round-Robin Allocation Algorithm}
\label{alg:rr}
\DontPrintSemicolon

\KwIn{Agents~$N$ and goods~$S_1$ where each good in~$S_1$ is considered as divisible by at most one agent.}
\KwOut{An EF1M and non-wasteful allocation of goods~$S_1$.}

$R \gets S_1$; $k \gets 0$\;
\While{$R \neq \emptyset$}{
    \tcp{Round~$k$}
    Increment the round counter to round~$k$.\;
    Initialize $N_k \gets N$ and $\pi_k \gets \emptyset$.\;
    \While{$N_k \neq \emptyset$}{
        Construct the graph~$G_k$ accordingly.\;
        Find in~$G_k$ either a cycle or a path in which the last vertex has no outgoing edge, and let~$S$ be the set of agents in this cycle or path.\; \label{alg:rr:S-property}
        Append all agents in~$S$ to the permutation~$\pi_k$ in an arbitrary order.\;
        For each~$i \in S$, assign~$g^i$ to agent~$i$.\;
        $N_k \gets N_k \setminus S$\;
        $R \gets R \setminus \{g^i \mid i \in S\}$\;
        Update~$N_k$ by removing agents who value every good in~$R$ at~$0$.
    }
    Update~$N$ by removing agents who value every good in~$R$ at~$0$.
}
\Return{the final allocation of goods~$S_1$}
\end{algorithm}

We are now ready to present the overall allocation algorithm, which is a generalized round-robin algorithm that proceeds in rounds.
In each round~$k$, it constructs the permutation~$\pi_k$ based on the repeatedly updated graph~$G_k$, and then lets the agents select their most preferred goods according to~$\pi_k$.
This continues until all goods are allocated.
The detailed steps are outlined in \Cref{alg:rr}.
To see why this algorithm satisfies EF1M and non-wastefulness, we show the following basic properties of the algorithm.

\begin{claim}
\label{claim:EF1M_nonwasteful_S1}
In each round~$k$, for each set~$S$ of agents selected in \cref{alg:rr:S-property} of Algorithm~\ref{alg:rr}, the following properties hold:
\begin{itemize}
\item \emph{Divisibility preservation:} For every agent $i \in S$, either $d(g^i) = \emptyset$ or $d(g^i) \in S$;
\item \emph{Good uniqueness:} For every distinct agents $i, j \in S$, $g^i \neq g^j$.
\end{itemize}
\end{claim}

\begin{proof}
For the first property, fix any~$i \in S$.
If $d(g^i) = \emptyset$, the statement holds.
We now consider the case where $d(g^i) \neq \emptyset$.
Note that the relevant graph~$G_k$ when agent~$i$ is considered by \cref{alg:rr:S-property} has exactly one directed edge from agent~$i$ to agent~$d(g^i)$.
Since \cref{alg:rr:S-property} finds either a cycle or a path in which the last vertex has no outgoing edge, $d(g^i)$ must be included in~$S$.

For the second property, suppose for contradiction that there exist distinct agents~$i, j \in S$ such that $g^i = g^j$.
Then, agents~$S$ form a directed graph in which agent~$d(g^i)$ has indegree of at least~$2$, violating the fact that our algorithm selects a cycle or a path in \cref{alg:rr:S-property}.
\end{proof}

These properties give us the following claim.

\begin{claim}
\label{claim:EF1M_nonwasteful_S2}
In each round~$k$, for any two agents~$i, j$, let~$g^i$ and~$g^j$ be the goods allocated to~$i$ and~$j$, respectively.
\begin{itemize}
\item Either $u_i(g^i) \geq u_i(g^j)$, or
\item Agent~$i$ considers~$g^j$ to be indivisible.
\end{itemize}
\end{claim}

\begin{proof}
If both agents~$i$ and~$j$ were added to~$S$ in the same execution of \cref{alg:rr:S-property}, then by the second property stated in \Cref{claim:EF1M_nonwasteful_S1} and the definition of~$g^i$, we have $u_i(g^i) \geq u_i(g^j)$.
In the following, we consider the scenario where agents~$i$ and~$j$ were added to~$S$ in two separate executions of \cref{alg:rr:S-property}.
Let us first consider that agent~$i$ was processed earlier than agent~$j$.
By the definition of~$g^i$, we have $u_i(g^i) \geq u_i(g^j)$.
We now consider that agent~$i$ was processed later than agent~$j$.
If $u_i(g^i) \geq u_i(g^j)$, the conclusion follows.
We thus assume that $u_i(g^i) < u_i(g^j)$ and will show that agent~$i$ considers~$g^j$ to be indivisible.
Suppose for contradiction that agent~$i$ considers~$g^j$ to be divisible, i.e., $d(g^j) = \{i\}$.
Consider the execution of \cref{alg:rr:S-property} when the algorithm processes agent~$j$.
Due to the construction of~$G_k$, there is a directed edge from agent~$j$ to agent~$i$.
By the first property of \Cref{claim:EF1M_nonwasteful_S1} and the fact that $d(g^j) \neq \emptyset$, agent~$i$ should have been added to~$S$ in the same execution of \cref{alg:rr:S-property} as agent~$j$, a contradiction.
\end{proof}

With these claims, we can now proceed to prove Lemma~\ref{lemma:EF1M_nonwasteful_S1}.

\begin{proof}[Proof of Lemma~\ref{lemma:EF1M_nonwasteful_S1}]
We prove the statement in two steps. First, we show that \Cref{alg:rr} successfully allocates all goods in $S_1$. Then, we prove by induction that each partial allocation produced by the algorithm satisfies EF1M.

We start by showing that when \Cref{alg:rr} terminates, all goods are allocated.
Given any constructed graph~$G_k$, \cref{alg:rr:S-property} can find a desired cycle or path by starting with any vertex on graph~$G_k$ and traversing the graph via outgoing edge of each visited vertex until
\begin{itemize}
\item either the last vertex has no outgoing edge, in which case we have found a desired path, or
\item we visit some visited vertex again, in which case we have found a cycle.
\end{itemize}
Clearly, as long as $N_k \neq \emptyset$, \cref{alg:rr:S-property} can always find a non-empty set of agents~$S$, meaning that in each inner \verb|while|-loop, a number of~$|S|$ goods are allocated.
Since we have assumed that every good is valued positively by some agent, all goods are allocated when the algorithm terminates.
Moreover, by the design of the algorithm, each good is allocated in its entirety to an agent who values it positively, non-wastefulness follows immediately.

In the following, we show the allocation of~$S_1$ satisfies \EFoneM.
For each round~$k$, let $\alloc^k = (A^k_1, \dots, A^k_n)$ denote the partial allocation after round~$k$ terminates.
We prove by induction that each partial allocation~$\alloc^k$ satisfies \EFoneM.
In the base case where $k = 0$, the empty allocation is \EFoneM.
Assume that the partial allocation~$\alloc^{k-1}$ is \EFoneM.
We will show that the partial allocation~$\alloc^k$ is \EFoneM.

Fix any pair of agents~$i, j \in N$.
Denote by~$\widehat{g}^i$ and~$\widehat{g}^j$ the two goods allocated to agents~$i$ and~$j$ in round~$k$, respectively.
(The good(s) could be empty set(s) if, say, some agent(s) value all the remaining goods~$R$ to be allocated in round~$k$ at~$0$.)
If $u_i(\widehat{g}^i) \geq u_i(\widehat{g}^j)$, it can be verified easily that agent~$i$ is still \EFoneM towards agent~$j$.
Else, $u_i(\widehat{g}^i) < u_i(\widehat{g}^j)$, and agent~$i$ considers~$\widehat{g}^j$ to be indivisible by \Cref{claim:EF1M_nonwasteful_S2}.
First, if $u_i(A^{k-1}_i) \geq u_i(A^{k-1}_j)$, allocation~$\alloc^k$ is \EFoneM because
\[
u_i(A^k_i) = u_i(A^{k-1}_i \cup \{\widehat{g}^i\}) \geq u_i(A^{k-1}_j) = u_i(A^k_j \setminus \{\widehat{g}^j\}).
\]
Next, we consider the case where $u_i(A^{k-1}_i) < u_i(A^{k-1}_j)$.
Since allocation~$\alloc^{k-1}$ is \EFoneM, there exists agent~$i$'s indivisible good~$g \in A^{k-1}_j$ such that $u_i(A^{k-1}_i) \geq u_i(A^{k-1}_j \setminus \{g\})$.
Let $k' \in [k-1]$ denote the first round such that $u_i(A^{k'}_i) < u_i(A^{k'}_j)$ and $u_i(A^{k'}_i) \geq u_i(A^{k'}_j \setminus \{g'\})$ for the good~$g'$ allocated to agent~$j$ in round~$k'$.
(Note that following from \Cref{claim:EF1M_nonwasteful_S2}, agent~$i$ considers good~$g'$ to be indivisible.)
By the design of the algorithm, from agent~$i$'s perspective, the value of the good allocated to agent~$i$ in round~$\ell = k', k'+1, \dots, k-1$ is at least that of the good allocated to agent~$j$ in round~$\ell + 1$, implying that $u_i(A^{k-1}_i \setminus A^{k'-1}_i) \geq u_i(A^k_j \setminus A^{k'}_j)$.
We thus have,
\[
u_i(A^k_i) = u_i(A^{k' - 1}_i) + u_i(A^{k-1}_i \setminus A^{k'-1}_i) + u_i(\widehat{g}^i) \geq u_i(A^{k' - 1}_j) + u_i(A^k_j \setminus A^{k'}_j) = u_i(A^k_j \setminus \{g'\}),
\]
showing allocation~$\alloc^k$ satisfies \EFoneM. Finally, the algorithm clearly runs in polynomial time.
The conclusion follows.
\end{proof}

Finally, the allocation returned by \Cref{alg:rr} can be extended by further allocating goods~$S_2$ while preserving both EF1M and non-wastefulness, completing the proof of Theorem~\ref{thm:EF1M_nonwasteful}.

\begin{proof}[Proof of Theorem~\ref{thm:EF1M_nonwasteful}]
Given any instance with agents~$N = [n]$ and goods~$M$, let~$S_1$ be the set of goods such that for each good in~$S_1$, there is at most one agent who considers the good to be divisible, and let $S_2 \coloneqq M \setminus S_1$ denote the other set of goods in which each good is considered as divisible for at least two agents.
By \Cref{lemma:EF1M_nonwasteful_S1}, we can find an allocation~$(A^1_i)_{i \in N}$ of goods~$S_1$ that is \EFoneM and non-wasteful in polynomial time.
It remains to allocate goods~$S_2$.
For each~$g \in S_2$, we divide good~$g$ equally among agents~$d(g)$.
Clearly, the allocation of~$S_2$, denoted as $(A^2_i)_{i \in N}$, is envy-free.

We now show the allocation $(A^1_i \cup A^2_i)_{i \in N}$ is \EFoneM and non-wasteful.
Non-wastefulness is clear, so we will focus on \EFoneM.
Fix any pair of agents~$i, j \in N$.
If, from agent~$i$'s perspective, agent~$j$ receives only divisible good from~$S_1$, we have $u_i(A^1_i) \geq u_i(A^1_j)$.
As a result,
\[
u_i(A^1_i \cup A^2_i) = u_i(A^1_i) + u_i(A^2_i) \geq u_i(A^1_j) + u_i(A^2_j) = u_i(A^1_j \cup A^2_j).
\]
If, from agent~$i$'s perspective, agent~$j$ receives some indivisible good(s) from~$S_1$, then, there exists agent~$i$'s indivisible good~$g \in A^1_j$ such that $u_i(A^1_i) \geq u_i(A^1_j \setminus \{g\})$.
Again,
\[
u_i(A^1_i \cup A^2_i) = u_i(A^1_i) + u_i(A^2_i) \geq u_i(A^1_j \setminus \{g\}) + u_i(A^2_j) = u_i(A^1_j \cup A^2_j \setminus \{g\}),
\]
as desired.
\end{proof}

What if we consider a stronger economic efficiency property like PO?
For mixed goods, every maximum Nash welfare (MNW) allocation satisfies \EFoneM and PO~\citep{CaragiannisKuMo19,NishimuraSu23}.
More formally, the \emph{Nash welfare} of an allocation~$(A_i)_{i \in N}$ is defined at $\prod_{i \in N} u_i(A_i)$.
An allocation~$(A_i)_{i \in N}$ is said to be a \emph{maximum Nash welfare (MNW)} allocation if it has the maximum Nash welfare among all allocations.\footnote{In the case where the maximum Nash welfare is~$0$, an allocation is MNW allocation if it gives positive utility to a set of agents of maximized size and moreover maximizes the product of utilities of the agents in that set.}
Perhaps surprisingly, this is not the case in our setting concerning subjective divisibility, even when there are only two agents.

\begin{example}[MNW does not imply \EFoneM]
Consider an instance with two agents and three goods with the following valuations and subjective divisibility, where $0 < \varepsilon < \frac{1}{2}$:
\begin{center}
\begin{tabular}{@{}l|lll@{}}
\toprule
& $g_1$ & $g_2$ & $g_3$ \\
\midrule
Agent~$1$ & ind, $1$ & div, $\frac{1}{2} + \varepsilon$ & div, $\frac{1}{2} + \varepsilon$ \\
Agent~$2$ & ind, $0$ & ind, $1$ & ind, $1$ \\
\bottomrule
\end{tabular}
\end{center}
It is clear that (i) in any MNW allocation, $g_1$ is given to agent~$1$, (ii) each of~$g_2$ and~$g_3$ is allocated fully to a single agent, and (iii) agent~$2$ gets at least one good.
Below, we discuss two allocations based on how we allocate~$g_2$ and~$g_3$.
\begin{itemize}
\item Allocation~$\alloc = (\{g_1\}, \{g_2, g_3\})$ has Nash welfare $1 \times (1 + 1) = 2$.

\item Allocation~$\alloc' = (\{g_1, g_2\}, \{g_3\})$ has Nash welfare $(1 + \frac{1}{2} + \varepsilon) \times 1 = \frac{3}{2} + \varepsilon$.

(Note that the allocation $(\{g_1, g_3\}, \{g_2\})$ is symmetric.)
\end{itemize}
Allocation~$\alloc$ is the MNW allocation, but does not satisfy \EFoneM as agent~$1$ thinks agent~$2$ receives only divisible good yet still envies agent~$2$.
\end{example}

It remains an interesting open question of whether \EFoneM and PO are compatible in our setting.

\section{Conclusion and Future Work}

In this paper, we have studied a fair division model where every agent has their own subjective divisibility towards the goods to be allocated.
We study fairness properties such as MMS and envy-freeness relaxations in our setting.
Regarding MMS, we show a tight $\frac{2}{3}$-approximation to MMS for $n \leq 3$ agents, and present a $\frac{1}{2}$-MMS algorithm for any number of agents.
Closing the gap between the algorithmic result and our impossibility result (which is~$2/3$) for $n \geq 4$ agents remains an intriguing and important open question.
For the hierarchy of envy-freeness relaxations, we show that \EFM is incompatible with non-wastefulness, even for two agents.
Nonetheless, if we discard at most one good, an \EFXM and non-wasteful allocation always exists for two agents.
When weakening the fairness property to \EFoneM, we show that, for any number of agents, an \EFoneM and non-wasteful allocation always exists.
The compatibility between \EFoneM and PO is a challenging and intriguing open question.

In \Cref{sec:extension-cake}, we extend our model to the case where agents' divisible goods are \emph{heterogeneous}.
Since homogeneous divisible goods are special cases, all of our impossibility results with respect to either MMS or EFM still hold.
We discuss how to adapt our algorithms for the cakes, and hence our algorithmic results carry over to this more general setting.

In future research, it would be interesting to further generalize our model of fair division with subjective divisibility to capture other practical scenarios, e.g., by imposing constraints on the allocation that can be made~\citep{Suksompong21}.
We have assumed each good to be \emph{completely} divisible or indivisible for the agents.
A natural next step is to allow more flexible subjective divisibility assumption, e.g., by having agents specify the degree of their subjective divisibility towards the goods.
One way to instantiate this assumption is via more complex utility functions.
To be more specific, we could let each agent impose a minimum threshold on each good such that they start getting positive utility by receiving (a subset of) the good.
Analogously, agents could also impose a maximum threshold on each good such that they will not get more utility by receiving more part of the good.

\section*{Acknowledgements}

A preliminary version (one-page abstract) of this paper appeared as~\citep{BeiLiLu23}.
We would like to thank the anonymous reviewers for their valuable feedback.

This work was partially supported by the Ministry of Education, Singapore, under its
Academic Research Fund Tier 1 (RG98/23), by ARC Laureate Project FL200100204 on ``Trustworthy AI'', by the National Natural Science Foundation of China (No.\ 62102117), by the Shenzhen Science and Technology Program (Nos.\ GXWD20231129111306002 and RCBS20210609103900003), by the Guangdong Basic and Applied Basic Research Foundation (No.\ 2023A1515011188), and by the CCF-Huawei Populus Grove Fund (No.\ CCF-HuaweiLK2022005).

\bibliographystyle{plainnat}
\bibliography{bibliography}

\begin{thebibliography}{35}
\providecommand{\natexlab}[1]{#1}
\providecommand{\url}[1]{\texttt{#1}}
\expandafter\ifx\csname urlstyle\endcsname\relax
  \providecommand{\doi}[1]{doi: #1}\else
  \providecommand{\doi}{doi: \begingroup \urlstyle{rm}\Url}\fi

\bibitem[Akrami and Garg(2024)]{AkramiGa24}
Hannaneh Akrami and Jugal Garg.
\newblock Breaking the $3/4$ barrier for approximate maximin share.
\newblock In \emph{Proceedings of the 35th ACM-SIAM Symposium on Discrete
  Algorithms (SODA)}, pages 74--91, 2024.

\bibitem[Akrami et~al.(2023)Akrami, Alon, Chaudhury, Garg, Mehlhorn, and
  Mehta]{AkramiAlCh23}
Hannaneh Akrami, Noga Alon, Bhaskar~Ray Chaudhury, Jugal Garg, Kurt Mehlhorn,
  and Ruta Mehta.
\newblock {EFX}: {A} simpler approach and an (almost) optimal guarantee via
  rainbow cycle number.
\newblock In \emph{Proceedings of the 24th ACM Conference on Economics and
  Computation (EC)}, page~61, 2023.

\bibitem[Alon(1987)]{Alon87}
Noga Alon.
\newblock Splitting necklaces.
\newblock \emph{Advances in Mathematics}, 63\penalty0 (3):\penalty0 247--253,
  1987.

\bibitem[Amanatidis et~al.(2023)Amanatidis, Aziz, Birmpas, Filos-Ratsikas, Li,
  Moulin, Voudouris, and Wu]{AmanatidisAzBi23}
Georgios Amanatidis, Haris Aziz, Georgios Birmpas, Aris Filos-Ratsikas, Bo~Li,
  Herv\'{e} Moulin, Alexandros~A. Voudouris, and Xiaowei Wu.
\newblock Fair division of indivisible goods: Recent progress and open
  questions.
\newblock \emph{Artificial Intelligence}, 322:\penalty0 103965, 2023.

\bibitem[Aziz et~al.(2016)Aziz, Bir\'{o}, Lang, Lesca, and Monnot]{AzizBiLa16}
Haris Aziz, P\'{e}ter Bir\'{o}, J\'{e}r\^{o}me Lang, Julien Lesca, and
  J\'{e}r\^{o}me Monnot.
\newblock Optimal reallocation under additive and ordinal preferences.
\newblock In \emph{Proceedings of the 15th International Conference on
  Autonomous Agents and Multiagent Systems (AAMAS)}, pages 402--410, 2016.

\bibitem[Bei et~al.(2021{\natexlab{a}})Bei, Li, Liu, Liu, and Lu]{BeiLiLi21}
Xiaohui Bei, Zihao Li, Jinyan Liu, Shengxin Liu, and Xinhang Lu.
\newblock Fair division of mixed divisible and indivisible goods.
\newblock \emph{Artificial Intelligence}, 293:\penalty0 103436,
  2021{\natexlab{a}}.

\bibitem[Bei et~al.(2021{\natexlab{b}})Bei, Liu, Lu, and Wang]{BeiLiLu21}
Xiaohui Bei, Shengxin Liu, Xinhang Lu, and Hongao Wang.
\newblock Maximin fairness with mixed divisible and indivisible goods.
\newblock \emph{Autonomous Agents and Multi-Agent Systems}, 35\penalty0
  (2):\penalty0 34:1--34:21, 2021{\natexlab{b}}.

\bibitem[Bei et~al.(2023)Bei, Liu, and Lu]{BeiLiLu23}
Xiaohui Bei, Shengxin Liu, and Xinhang Lu.
\newblock Fair division with subjective divisibility.
\newblock In \emph{Proceedings of the 19th Conference on Web and Internet
  Economics (WINE)}, page 677, 2023.

\bibitem[Berger et~al.(2022)Berger, Cohen, Feldman, and Fiat]{BergerCoFe22}
Ben Berger, Avi Cohen, Michal Feldman, and Amos Fiat.
\newblock Almost full {EFX} exists for four agents.
\newblock In \emph{Proceedings of the 36th AAAI Conference on Artificial
  Intelligence (AAAI)}, pages 4826--4833, 2022.

\bibitem[Bhaskar et~al.(2021)Bhaskar, Sricharan, and Vaish]{BhaskarSrVa21}
Umang Bhaskar, A.~R. Sricharan, and Rohit Vaish.
\newblock On approximate envy-freeness for indivisible chores and mixed
  resources.
\newblock In \emph{Proceedings of the 24th International Conference on
  Approximation Algorithms for Combinatorial Optimization Problems (APPROX)},
  pages 1:1--1:23, 2021.

\bibitem[Bouveret and Lema\^{i}tre(2016)]{BouveretLe16}
Sylvain Bouveret and Michel Lema\^{i}tre.
\newblock Characterizing conflicts in fair division of indivisible goods using
  a scale of criteria.
\newblock \emph{Autonomous Agents and Multi-Agent Systems}, 30\penalty0
  (2):\penalty0 259--290, 2016.

\bibitem[Budish(2011)]{Budish11}
Eric Budish.
\newblock The combinatorial assignment problem: Approximate competitive
  equilibrium from equal incomes.
\newblock \emph{Journal of Political Economy}, 119\penalty0 (6):\penalty0
  1061--1103, 2011.

\bibitem[Caragiannis et~al.(2019)Caragiannis, Kurokawa, Moulin, Procaccia,
  Shah, and Wang]{CaragiannisKuMo19}
Ioannis Caragiannis, David Kurokawa, Herv\'{e} Moulin, Ariel~D. Procaccia,
  Nisarg Shah, and Junxing Wang.
\newblock The unreasonable fairness of maximum {N}ash welfare.
\newblock \emph{ACM Transactions on Economics and Computation}, 7\penalty0
  (3):\penalty0 12:1--12:32, 2019.

\bibitem[Chaudhury et~al.(2021)Chaudhury, Kavitha, Mehlhorn, and
  Sgouritsa]{ChaudhuryKaMe21}
Bhaskar~Ray Chaudhury, Telikepalli Kavitha, Kurt Mehlhorn, and Alkmini
  Sgouritsa.
\newblock A little charity guarantees almost envy-freeness.
\newblock \emph{SIAM Journal on Computing}, 50\penalty0 (4):\penalty0
  1336--1358, 2021.

\bibitem[Chaudhury et~al.(2024)Chaudhury, Garg, and Mehlhorn]{ChaudhuryGaMe24}
Bhaskar~Ray Chaudhury, Jugal Garg, and Kurt Mehlhorn.
\newblock {EFX} exists for three agents.
\newblock \emph{Journal of the ACM}, 71\penalty0 (1):\penalty0 4:1--4:27, 2024.

\bibitem[de~Keijzer et~al.(2009)de~Keijzer, Bouveret, Klos, and
  Zhang]{deKeijzerBoKl09}
Bart de~Keijzer, Sylvain Bouveret, Tomas Klos, and Yingqian Zhang.
\newblock On the complexity of efficiency and envy-freeness in fair division of
  indivisible goods with additive preferences.
\newblock In \emph{Proceedings of the 1st International Conference on
  Algorithmic Decision Theory (ADT)}, pages 98--110, 2009.

\bibitem[Farhadi et~al.(2019)Farhadi, Ghodsi, Hajiaghayi, Lahaie, Pennock,
  Seddighin, Seddighin, and Yami]{FarhadiGhHa19}
Alireza Farhadi, Mohammad Ghodsi, MohammadTaghi Hajiaghayi, S\'{e}bastien
  Lahaie, David Pennock, Masoud Seddighin, Saeed Seddighin, and Hadi Yami.
\newblock Fair allocation of indivisible goods to asymmetric agents.
\newblock \emph{Journal of Artificial Intelligence Research}, 64\penalty0
  (1):\penalty0 1--20, 2019.

\bibitem[Feige and Norkin(2022)]{FeigeNo22}
Uriel Feige and Alexey Norkin.
\newblock Improved maximin fair allocation of indivisible items to three
  agents.
\newblock \emph{CoRR}, abs/2205.05363, 2022.

\bibitem[Feige et~al.(2021)Feige, Sapir, and Tauber]{FeigeSaTa21}
Uriel Feige, Ariel Sapir, and Laliv Tauber.
\newblock A tight negative example for {MMS} fair allocations.
\newblock In \emph{Proceedings of the 17th Conference on Web and Internet
  Economics (WINE)}, pages 355--372, 2021.

\bibitem[Foley(1967)]{Foley67}
Duncan~Karl Foley.
\newblock Resource allocation and the public sector.
\newblock \emph{Yale Economics Essays}, 7\penalty0 (1):\penalty0 45--98, 1967.

\bibitem[Kurokawa et~al.(2018)Kurokawa, Procaccia, and Wang]{KurokawaPrWa18}
David Kurokawa, Ariel~D. Procaccia, and Junxing Wang.
\newblock Fair enough: Guaranteeing approximate maximin shares.
\newblock \emph{Journal of the ACM}, 65\penalty0 (2):\penalty0 8:1--8:27, 2018.

\bibitem[Li et~al.(2023)Li, Liu, Lu, and Tao]{LiLiLu23}
Zihao Li, Shengxin Liu, Xinhang Lu, and Biaoshuai Tao.
\newblock Truthful fair mechanisms for allocating mixed divisible and
  indivisible goods.
\newblock In \emph{Proceedings of the 32nd International Joint Conference on
  Artificial Intelligence (IJCAI)}, pages 2808--2816, 2023.

\bibitem[Li et~al.(2024)Li, Liu, Lu, Tao, and Tao]{LiLiLu24}
Zihao Li, Shengxin Liu, Xinhang Lu, Biaoshuai Tao, and Yichen Tao.
\newblock A complete landscape for the price of envy-freeness.
\newblock In \emph{Proceedings of the 23rd International Conference on
  Autonomous Agents and Multiagent Systems (AAMAS)}, pages 1183--1191, 2024.

\bibitem[Lipton et~al.(2004)Lipton, Markakis, Mossel, and Saberi]{LiptonMaMo04}
Richard~J. Lipton, Evangelos Markakis, Elchanan Mossel, and Amin Saberi.
\newblock On approximately fair allocations of indivisible goods.
\newblock In \emph{Proceedings of the 5th ACM Conference on Electronic Commerce
  (EC)}, pages 125--131, 2004.

\bibitem[Liu et~al.(2024)Liu, Lu, Suzuki, and Walsh]{LiuLuSu24}
Shengxin Liu, Xinhang Lu, Mashbat Suzuki, and Toby Walsh.
\newblock Mixed fair division: {A} survey.
\newblock \emph{Journal of Artificial Intelligence Research}, 80:\penalty0
  1373--1406, 2024.

\bibitem[Misra and Sethia(2020)]{MisraSe20}
Neeldhara Misra and Aditi Sethia.
\newblock Fair division is hard even for amicable agents.
\newblock In \emph{Proceedings of the 21st Italian Conference on Theoretical
  Computer Science (ICTCS)}, pages 202--207, 2020.

\bibitem[Nishimura and Sumita(2023)]{NishimuraSu23}
Koichi Nishimura and Hanna Sumita.
\newblock Envy-freeness and maximum {N}ash welfare for mixed divisible and
  indivisible goods.
\newblock \emph{CoRR}, abs/2302.13342, 2023.

\bibitem[Plaut and Roughgarden(2020)]{PlautRo20}
Benjamin Plaut and Tim Roughgarden.
\newblock Almost envy-freeness with general valuations.
\newblock \emph{SIAM Journal on Discrete Mathematics}, 34\penalty0
  (2):\penalty0 1039--1068, 2020.

\bibitem[Procaccia(2016)]{Procaccia16}
Ariel~D. Procaccia.
\newblock Cake cutting algorithms.
\newblock In Felix Brandt, Vincent Conitzer, Ulle Endriss, J\'{e}r\^{o}me Lang,
  and Ariel~D. Procaccia, editors, \emph{Handbook of Computational Social
  Choice}, chapter~13, pages 311--329. Cambridge University Press, 2016.

\bibitem[Robertson and Webb(1998)]{RobertsonWe98}
Jack Robertson and William Webb.
\newblock \emph{Cake-Cutting Algorithm: Be Fair If You Can}.
\newblock A K Peters/CRC Press, 1998.

\bibitem[Sandomirskiy and Segal-Halevi(2022)]{SandomirskiySe22}
Fedor Sandomirskiy and Erel Segal-Halevi.
\newblock Efficient fair division with minimal sharing.
\newblock \emph{Operations Research}, 70\penalty0 (3):\penalty0 1762--1782,
  2022.

\bibitem[Steinhaus(1949)]{Steinhaus49}
Hugo Steinhaus.
\newblock Sur la division pragmatique.
\newblock \emph{Econometrica}, 17:\penalty0 315--319, 1949.

\bibitem[Suksompong(2021)]{Suksompong21}
Warut Suksompong.
\newblock Constraints in fair division.
\newblock \emph{SIGecom Exchanges}, 19\penalty0 (2):\penalty0 46--61, 2021.

\bibitem[Suksompong(2025)]{Suksompong25}
Warut Suksompong.
\newblock Weighted fair division of indivisible items: {A} review.
\newblock \emph{Information Processing Letters}, 187:\penalty0 106519, 2025.

\bibitem[Woeginger(1997)]{Woeginger97}
Gerhard~J. Woeginger.
\newblock A polynomial-time approximation scheme for maximizing the minimum
  machine completion time.
\newblock \emph{Operations Research Letters}, 20\penalty0 (4):\penalty0
  149--154, 1997.

\end{thebibliography}

\appendix

\section{Heterogeneous Divisible Goods (i.e., Cakes)}
\label{sec:extension-cake}

We consider an extension of our model where agents' divisible goods are \emph{heterogeneous}.
In detail, as is the case in \Cref{sec:preliminaries}, each agent~$i \in N$ can partition goods~$M = \{g_1, g_2, \dots, g_m\}$ into two disjoint subsets~$\indGoodsOf{i}$ and~$\divGoodsOf{i}$; however, now, $\divGoodsOf{i}$ consists of the agent's heterogeneous divisible goods (i.e., \emph{cakes}).
Specifically, each agent~$i \in N$ is endowed with a \emph{density function}~$f_{ij} \colon [j-1, j] \to \mathbb{R}_{> 0}$ for each divisible good~$g_j \in \divGoodsOf{i}$, which captures how the agent values different parts of her divisible good~$g_j$.\footnote{Most cake cutting papers assume the density function to be \emph{non-negative}; we do not follow this convention.}
It is worth mentioning that our main model in \Cref{sec:preliminaries} corresponds to the special case where each agent~$i$'s density function for her divisible good~$g_j \in \divGoodsOf{i}$ being a constant~$c_{ij}$ such that $\int_{j-1}^{j} c_{ij} \diff{x} = u_i(g_j)$.
We now describe the cake-cutting query model of \citet{RobertsonWe98} (shortly, RW model), which allows algorithms to interact with the agent via the following two types of queries to get access to agents' density functions:
\begin{itemize}
\item $\textsc{Eval}_i(x, y)$ asks agent~$i$ to evaluate the interval~$[x, y]$ and returns the value~$u_i([x, y])$;

\item $\textsc{Cut}_i(x, \alpha)$ asks agent~$i$ to return the leftmost point~$y$ such that $u_i([x, y]) = \alpha$, or state that no such point exists.
\end{itemize}
We still assume additive utilities, i.e., given an allocation~$\alloc = (A_1, A_2, \dots, A_n)$, agent~$i$'s utility is
\[
u_i(\alloc) = u_i(A_i) = \sum_{g \in M} \left( u_i(g) \cdot \indicator{g \in \indGoodsOf{i}\, \land\, g^{A_i} = g} + \int_{g^{A_i}} f_{ig} \diff{x} \cdot \indicator{g \in \divGoodsOf{i}} \right);
\]
recall that~$g^{A_i}$ denotes the part of good~$g$ that is contained in bundle~$A_i$.
Clearly, all of our fairness notions, MMS and envy-freeness relaxations, as well as the economic efficiency notions are well-defined in the current setting.

\paragraph{Impossibility Results}
Since homogeneous divisible goods are special cases of cakes, all of our impossibility results still hold.
In sum, we have the following corollaries.

\begin{corollary}
In the $n$-agent case, when agents' divisible goods are heterogeneous,
\begin{itemize}
\item the worst-case MMS approximation guarantee is at most~$2/3$;
\item \EFM and non-wastefulness are incompatible.
\end{itemize}
\end{corollary}

\paragraph{Positive Results}
We now proceed to discuss our algorithmic results obtained in \Cref{sec:MMS,sec:EFM}, beginning with MMS.
We first modify the \highValuedAlg (\Cref{alg:high-valued}) as it is used in the $\frac{2}{3}$-MMS algorithms for two and three agents (\Cref{alg:2-agent-2/3-MMS,alg:2/3-MMS}, respectively) and the $\frac{1}{2}$-MMS algorithm for $n$ agents (\Cref{alg:n-agent-1/2-MMS}).
The modification is straightforward: in \cref{ALG-n:high-valued:cut} of \Cref{alg:high-valued}, we now let the algorithm interact with each agent~$i \in N$ via the ``$\textsc{Cut}_i(x, \beta \cdot \MMS_i)$'' query, where $x$ is the leftmost position of good~$g$, so that the agents claim feasible pieces of the least length.
It is easy to verify that \Cref{lem:n-agent-high-valued} still holds with cakes.
Next, all results revolved around an integral bundle also hold with cakes as what matters here is how agents value each entire good.
In short, \Cref{alg:2-agent-2/3-MMS} outputs a $\frac{2}{3}$-MMS allocation for two agents when they have cakes.
In terms of the $3$-agent-$\frac{2}{3}$-MMS algorithm, due to the above arguments, the only step we need to pay extra attention is \cref{ALG-3:common-div:cut} of \Cref{alg:2/3-MMS}, and the modification here is also simple: impose ``$\textsc{Cut}_i(x, \frac{\utility{i}{\{g_1, d, g_2\}}}{2})$'' query, where $x$ is the leftmost position of the line of goods~$g_1, d, g_2$.
Our discussion in \Cref{sec:MMS:discussion} about the computation of agents' MMS as well as the considerations of economic efficiency also applies to the current setting.

We now consider envy-freeness relaxations.
With our assumption that agents' density functions are \emph{strictly positive}, it is easy to verify that \Cref{alg:2-agent-charity} outputs an \EFXM and non-wasteful allocation in polynomial time for two agents when they have heterogeneous divisible goods if at most one good is discarded.
Similarly, \Cref{thm:EF1M_nonwasteful} also holds for cakes, with the following caveat: when allocating the goods each of which is divisible for at least two agents, we use the \emph{perfect allocation} of \citet{Alon87} to ensure that each agent equally values all pieces to be allocated---such an allocation, however, cannot be implemented with a finite number of queries in the RW model~\citep{RobertsonWe98}.

\section{Omitted Details}

\subsection{Approximate MMS Values}
\label{app:MMS-values}

We now give the details of how to efficiently approximate agents' MMS values.
Fix any agent~$i \in N$.
Let agent~$i$ divide her divisible goods $\divGoodsOf{i}$ into at most $\lceil \frac{2n}{\varepsilon} \rceil + |\divGoodsOf{i}| - 1$ intervals each of which is worth at most $\frac{\varepsilon \cdot u_i(\divGoodsOf{i})}{2n}$.
Since $\frac{u_i(\divGoodsOf{i})}{n} \leq \MMS_i(n, M)$, each interval is also worth at most $\frac{\varepsilon}{2} \cdot \MMS_i(n, M)$.
Denote by $\widehat{M}$ the collection of these discretized, \emph{indivisible} intervals, and by $\langle [n], \indGoodsOf{i} \cup \widehat{M} \rangle$ the created instance containing only indivisible goods.
By applying the PTAS of \citet{Woeginger97}, for any constant $\delta > 0$, we can compute a partition $\mathcal{P} = (P_1, P_2, \dots, P_n)$ of goods $\indGoodsOf{i} \cup \widehat{M}$ such that
\[
\min_{j \in [n]} u_i(P_j) \geq (1 - \delta) \cdot \MMS_i(n, \indGoodsOf{i} \cup \widehat{M}).
\]
We will show below that $\MMS_i(n, \indGoodsOf{i} \cup \widehat{M}) \geq \left( 1- \frac{\varepsilon}{2} \right) \cdot \MMS_i(n, M)$.
By setting~$\delta = \varepsilon / 2$,
\begin{align*}
\min_{j \in [n]} u_i(P_j) &\geq (1 - \delta) \cdot \MMS_i(n, \indGoodsOf{i} \cup \widehat{M}) \\
&\geq \left( 1 - \frac{\varepsilon}{2} \right) \cdot \left( 1- \frac{\varepsilon}{2} \right) \cdot \MMS_i(n, M) \geq (1 - \varepsilon) \cdot \MMS_i(n, M).
\end{align*}
To summarize, this new PTAS described above computes $(1 - \varepsilon) \cdot \MMS_i(n, M)$.

In what follows, we establish the inequality $\MMS_i(n, \indGoodsOf{i} \cup \widehat{M}) \geq \left( 1- \frac{\varepsilon}{2} \right) \cdot \MMS_i(n, M)$.
It is worth noting that the following constructive argument is not part of the aforementioned new PTAS.
At a high level, we use agent~$i$'s MMS partition~$\mathcal{S}$ of goods $M = \indGoodsOf{i} \cup \divGoodsOf{i}$ as a guidance to construct a partition~$\mathcal{T} = (T_1, T_2, \dots, T_n)$ of goods $\indGoodsOf{i} \cup \widehat{M}$ for agent~$i$:
\begin{itemize}
\item Let the partition of good~$\indGoodsOf{i}$ in~$\mathcal{T}$ be exactly the same as that in~$\mathcal{S}$.
\item For any bundle in~$\mathcal{T}$ that is worth less than $\left( 1 - \frac{\varepsilon}{2} \right) \cdot \MMS_i(n, M)$, we add one interval from~$\widehat{M}$ at a time to this bundle until its value falls in $\left[ \left( 1 - \frac{\varepsilon}{2} \right) \cdot \MMS_i(n, M), \MMS_i(n, M) \right]$.

Recall that each interval is worth at most $\frac{\varepsilon}{2} \cdot \MMS_i(n, M)$, so this operation is valid.

There are also enough intervals for this operation as in partition~$\mathcal{S}$, each bundle is worth at least $\MMS_i(n, M)$.

\item We then distribute all remaining intervals to any bundles in partition~$\mathcal{T}$.
\end{itemize}
This is the desired partition~$\mathcal{T}$ we would like to construct, as
\[
\MMS_i(n, \indGoodsOf{i} \cup \widehat{M}) \geq \min_{j \in [n]} u_i(T_j) \geq \left( 1- \frac{\varepsilon}{2} \right) \cdot \MMS_i(n, M).
\]
We would like to remark again our argument involving the construction of partition~$\mathcal{T}$ is only used to prove the inequality.

\subsection{Proof of Theorem~\ref{thm:EFM-non-wasteful-incompatible} (Cont.)}
\label{app:omitted-proofs}

Consider an instance with subjective divisibility and utilities as shown below, where $0 < \varepsilon < 1/2$.
Specifically, good~$g_0$ is a unanimous indivisible good which is only positively valued by agents~$1$ and~$2$.
For each~$i \in [n]$, good~$g_i$ is divisible for agent~$i$ but indivisible for all other agents.
For each~$i \in \{3, 4, \dots, n\}$, agent~$i$ only positively values good~$g_i$.

\begin{center}
\begin{tabular}{@{}c|*{8}{l}@{}}
\toprule
& $g_0$ & $g_1$ & $g_2$ & $g_3$ & \ldots & $g_i$ & \ldots & $g_n$\\
\midrule
Agent~$1$ & ind, $1 - \varepsilon/2$ & div, $\varepsilon$ & ind, $1 - \varepsilon$ & ind, $1$ & \ldots & ind, $1$ & \ldots & ind, $1$ \\
Agent~$2$ & ind, $1 - \varepsilon/2$ & ind, $1 - \varepsilon$ & div, $\varepsilon$ & ind, $1$ & \ldots & ind, $1$ & \ldots & ind, $1$ \\
Agent~$3$ & ind, $0$ & ind, $0$ & ind, $0$ & div, $1$ & \ldots & ind, $0$ & \ldots & ind, $0$ \\
\vdots & \vdots & \vdots & \vdots & \vdots & $\ddots$ & \vdots & $\ddots$ & \vdots \\
Agent~$i$ & ind, $0$ & ind, $0$ & ind, $0$ & ind, $0$ & \ldots & div, $1$ & \ldots & ind, $0$ \\
\vdots & \vdots & \vdots & \vdots & \vdots & $\ddots$ & \vdots & $\ddots$ & \vdots \\
Agent~$n$ & ind, $0$ & ind, $0$ & ind, $0$ & ind, $0$ & \ldots & ind, $0$ & \ldots & div, $1$ \\
\bottomrule
\end{tabular}
\end{center}

First, non-wastefulness requires all goods to be allocated as a whole, and thus in the following, we only focus on integral allocations.
Next, in any EFM allocation, for each~$i \in \{3, 4, \dots, n\}$, agent~$i$ should get good~$g_i$; otherwise, agent~$i$ would envy the agent who gets good~$g_i$ and thus not be EFM.
Last, we discuss how goods~$g_0, g_1, g_2$ are distributed.
It is clear that in any EFM allocation, neither agent~$1$ nor agent~$2$ can get an empty bundle.
This is because good~$g_1$ (resp., good~$g_2$) is given to some agent and agent~$1$ (resp., agent~$2$) should be envy-free towards this agent.
If goods~$g_0, g_1, g_2$ are distributed between agents~$1$ and~$2$, according to argument on \cpageref{thm:EFM-non-wasteful-incompatible}, EFM is not attainable.
So, we may try to let some agent~$j \in \{3, 4, \dots, n\}$ get exactly one good from $\{g_0, g_1, g_2\}$; note that neither of agents~$1$ and~$2$ gets an empty bundle.
Such a redistribution, however, cannot give an EFM allocation as well.
For instance, consider the allocation before redistribution (\Cref{table:redistribution-not-EFM}) in which agents~$1$ and~$2$ get bundles~$\{g_0, g_1\}$ and~$\{g_2\}$, respectively:
\begin{itemize}
\item If good~$g_0$ is given to agent~$j$, then agent~$1$ is not EF1 towards agent~$j$ because $u_1(g_1) = \varepsilon$ is less than $u_1(g_0) = 1 - \varepsilon/2$ when good~$g_j$ is removed from agent~$j$'s bundle~$\{g_0, g_j\}$.
Similar argument works for agent~$2$.

\item Else, good~$g_1$ is given to agent~$j$, then agent~$1$ is not envy-free towards agent~$j$ due to $u_1(g_0) = 1 - \varepsilon/2 < u_1(\{g_1, g_j\}) = 1 + \varepsilon$ and agent~$2$ is not EF1 towards agent~$j$ because $u_2(g_2) = \varepsilon < u_2(g_1) = 1 - \varepsilon$ when good~$g_j$ is removed from agent~$j$'s bundle~$\{g_1, g_j\}$.
\end{itemize}
All other cases are listed in \Cref{table:redistribution-not-EFM}.

\begin{table}[t]
\centering
\begin{tabular}{@{}ll|l|l@{}}
\toprule
Agent~$1$ & Agent~$2$ & \multicolumn{2}{l}{Remarks} \\
\midrule
$\{g_0, g_1\}$ & $\{g_2\}$ & $j \gets g_0$: $a_1, a_2$ are not EF1. & $j \gets g_1$: $a_1$ is not EF; $a_2$ is not EF1. \\
$\{g_2\}$ & $\{g_0, g_1\}$ & $j \gets g_0$: $a_1, a_2$ are not EF1. & $j \gets g_1$: $a_1$ is not EF. \\
$\{g_0, g_2\}$ & $\{g_1\}$ & $j \gets g_0$: $a_1, a_2$ are not EF1. & $j \gets g_2$: $a_2$ is not EF. \\
$\{g_1\}$ & $\{g_0, g_2\}$ & $j \gets g_0$: $a_1, a_2$ are not EF1. & $j \gets g_2$: $a_1$ is not EF1; $a_2$ is not EF. \\
$\{g_1, g_2\}$ & $\{g_0\}$ & $j \gets g_1$: $a_1$ is not EF. & $j \gets g_2$: $a_1$ is not EF1; $a_2$ is not EF. \\
$\{g_0\}$ & $\{g_1, g_2\}$ & $j \gets g_1$: $a_1$ is not EF; $a_2$ is not EF1. & $j \gets g_2$: $a_2$ is not EF. \\
\bottomrule
\end{tabular}
\caption{Redistribution of~$g_0, g_1, g_2$ does not give an EFM allocation in the proof of \Cref{thm:EFM-non-wasteful-incompatible}.}
\label{table:redistribution-not-EFM}
\end{table}

We thus conclude that EFM and non-wastefulness are incompatible if all goods need to be allocated.
\end{document}